\documentclass[journal]{IEEEtran}
\IEEEoverridecommandlockouts
\usepackage{graphicx, subcaption}
\usepackage{amsmath,amsthm,amsfonts,amssymb}
\usepackage{cite, soul}
\usepackage{url}
\usepackage{bm}
\usepackage{comment,enumitem}

\DeclareMathOperator*{\argmin}{arg\,min}
\usepackage[usenames,dvipsnames]{color}
\theoremstyle{plain}
\newtheorem{thm}{Theorem}

\newtheorem{prop}[thm]{Proposition}

\newtheorem*{rem}{Remark}

\title{Towards AI-Native Fronthaul: Neural Compression for NextG Cloud RAN}

\author{
Chenghong~Bian, \IEEEmembership{Student Member,~IEEE},
Yulin~Shao, \IEEEmembership{Member,~IEEE},
Deniz~G{\"u}nd{\"u}z, \IEEEmembership{Fellow,~IEEE}
\thanks{C. Bian and D. G\"und\"uz are with the Department of Electrical and Electronic Engineering, Imperial College London (E-mails: \{c.bian22, d.gunduz\}@imperial.ac.uk).}
\thanks{Y. Shao is with the Department of Electrical and Electronic Engineering, The University of Hong Kong (E-mail: ylshao@hku.hk).}
}

\begin{document}

\maketitle
\begin{abstract}
The rapid growth of data traffic and the emerging AI-native wireless architectures in NextG cellular systems place new demands on the fronthaul links of Cloud Radio Access Networks (C-RAN). In this paper, we investigate neural compression techniques for the Common Public Radio Interface (CPRI), aiming to reduce the fronthaul bandwidth while preserving signal quality.
We introduce two deep learning-based compression algorithms designed to optimize the transformation of wireless signals into bit sequences for CPRI transmission. The first algorithm utilizes a non-linear transformation coupled with scalar/vector quantization based on a learned codebook. The second algorithm generates a latent vector transformed into a variable-length output bit sequence via arithmetic encoding, guided by the predicted probability distribution of each latent element. Novel techniques such as a shared weight model for storage-limited devices and a successive refinement model for managing multiple CPRI links with varying Quality of Service (QoS) are proposed. Extensive simulation results demonstrate notable Error Vector Magnitude (EVM) gains with improved rate-distortion performance for both algorithms compared to traditional methods. 
The proposed solutions are robust to variations in channel conditions, modulation formats, and noise levels, highlighting their potential for enabling efficient and scalable fronthaul in NextG AI-native networks, aligning with the current 3GPP research directions.
\end{abstract}

\begin{IEEEkeywords}
C-RAN, AI-native fronthaul, CPRI compression, vector quantization, VQ-VAE, variable-rate compression.
\end{IEEEkeywords}

\section{Introduction}
\label{sec:intro}

{ As mobile networks evolve toward the 6G and NextG era, characterized by ubiquitous intelligence and ultra-dense deployments, the pressure on fronthaul capacity becomes increasingly critical. In particular, the demand for real-time, high-throughput transmission between distributed radio units and centralized baseband processing units, typified in Cloud Radio Access Network (C-RAN) architectures, poses significant challenges \cite{CRAN1}. C-RAN enhances spectral and energy efficiency by decoupling Remote Radio Heads (RRHs) from Baseband Units (BBUs), but relies on high-speed fronthaul links, such as the Common Public Radio Interface (CPRI) \cite{CPRI1,CPRI2}, to transport raw physical-layer signals. Without effective compression, this data exchange can quickly overwhelm fronthaul resources, hindering scalability and performance in future-generation systems. Thus, CPRI compression emerges as a key enabler for scalable, efficient C-RAN deployments aligned with 6G and NextG network design.}

In advancing the efficiency of signal transmission over the CPRI link, significant strides have been made through various research efforts. At the core of these innovations is the development of sophisticated CPRI compression techniques aimed at enhancing transmission efficiency and signal quality.
The foundational approach \cite{scalar_cpri} focuses on exploiting the inherent redundancy within time-domain signals, wherein the authors devise a method incorporating multi-rate filtering to pare down the numbers of samples to be transmitted. This process is complemented by the application of non-uniform scalar quantization, effectively translating each I/Q sample into a compressed bit representation, marking a pivotal step forward in CPRI link optimization.

\begin{figure}[t]
     \centering
     \begin{subfigure}{\columnwidth}
         \centering
         \includegraphics[width=0.75\columnwidth]{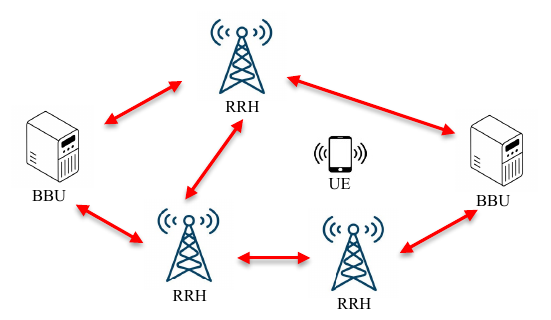}
         \caption{The connection of RRHs to BBUs via CPRI links within a C-RAN architecture.}

     \end{subfigure}     
     %\vspace{1cm}
     \begin{subfigure}{\columnwidth}
         \centering
         \includegraphics[width=\columnwidth]{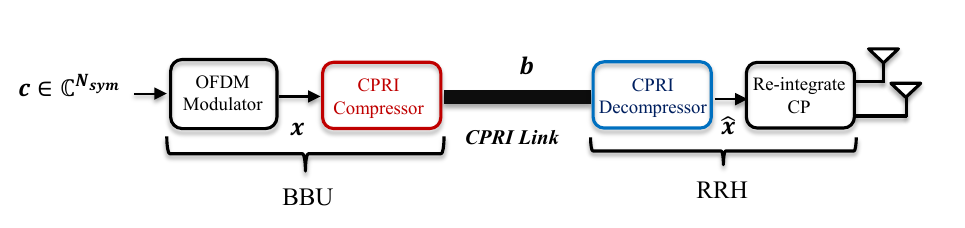}
         \caption{The processing of the BBU and RRH in the downlink scenario.}

     \end{subfigure} 
     
     \begin{subfigure}{\columnwidth}
         \centering
         \includegraphics[width=\columnwidth]{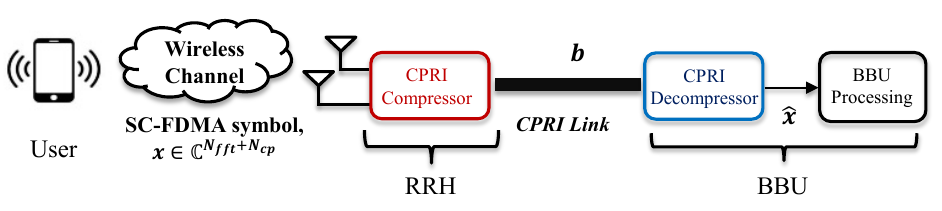}
         \caption{The uplink scenario.}

     \end{subfigure}

  \caption{Illustration of the CPRI link for both downlink and uplink scenarios.}
  \label{fig:fig_system}
\end{figure}

Building upon this groundwork, \cite{vq_cpri} introduces a vector quantization scheme. By organizing the input sequence into vectors, and mapping these to the nearest codewords within a specially curated codebook -- derived through a modified Lloyd algorithm -- this method significantly improves compression efficiency. Additionally, the integration of entropy coding further enhances this system, yielding better reconstruction performance and a notable advancement over the scalar quantization approach.
A filter with an error-feedback loop is deployed in \cite{noise_shape_cpri}, resulting in even more impressive compression performance. This approach relies on element-wise quantization of filter's output signal, which ingeniously shifts quantization noise away from OFDM subcarriers carrying vital information to those that do not, thereby improving reconstruction performance.

Expanding the scope of CPRI compression, researchers have also explored its application in distributed systems \cite{distributed_cpri, meta_cpri} and multiple-input multiple-output (MIMO) scenarios \cite{mimo_cpri}. These studies underscore the versatility and critical importance of CPRI compression techniques in refining the functionality and efficiency of C-RAN systems, setting a precedent for future innovations in 5G-Advanced and beyond.

% Efforts have been made by the researchers to facilitate signal transmission over the CPRI link. The seminal work is described in \cite{scalar_cpri} where the authors identify the redundancy of the time-domain signal and design a multi-rate filter to reduce the number of samples. They then perform non-uniform scalar quantization to transform each I/Q sample to bit representation. 
% The authors in \cite{vq_cpri} improve the work in \cite{scalar_cpri} by considering a vector quantization scheme. In particular, the input sequence is grouped into vectors and each vector will be represented by a codeword from a given codebook which is nearest to it in terms of L2 distance. The codebook is obtained by a modified Lloyd algorithm. The authors also consider entropy coding leading to better compression efficiency and an improved reconstruction performance is reported compared to \cite{scalar_cpri}.
% It is found in \cite{noise_shape_cpri} that superior compression performance can be obtained by designing a filter with error-feedback loop. The output signal of the filter is quantized element-wise for transmission. Better reconstruction performance is observed since the filter moves the quantization noise located at all the OFDM subcarriers to the ones which do not carry information. 
% Distributed CPRI compression and compression in the MIMO scenario are considered in \cite{distributed_cpri} and \cite{mimo_cpri}, respectively. 

The integration of deep learning (DL) into wireless communications in recent years has heralded a significant paradigm shift, demonstrating its efficacy in enhancing various aspects of the field, including signal processing \cite{mimo_det, dl_amp}, source coding \cite{balle2016end, balle2018variational}, channel coding \cite{dl_phy,GBAFC}, and joint source-channel coding \cite{deepjscc}. This pioneering research has effectively leveraged domain-specific insights to refine deep neural network (DNN) architectures and their training methodologies, underscoring the versatility and power of deep learning technologies in addressing complex communication challenges.

Driven by the substantial advancements facilitated by DL in wireless communications, researchers have developed machine learning based approach for CPRI compression \cite{meta_cpri, drl_cpri, deep_cpri}. 
In particular, the authors in \cite{meta_cpri} consider a scenario where the user is served by a central BBU cooperatively via multiple RRHs. Different transformation matrices at the RRHs are learned to perform dimension reduction/expansion followed by uniform quantization.
In \cite{drl_cpri}, the authors optimize the CPRI compression problem using a deep reinforcement learning approach under the latency and packet loss constraints. In our case, we consider the user is served by the BBU connected to a single RRH and utilize the EVM as the performance metric.
The most relevant work is  \cite{deep_cpri} which applies a DNN for non-linear transformation on a per-element basis, essentially functioning as a non-uniform scalar quantizer and achieves only marginal improvements over traditional methods, such as those described in \cite{scalar_cpri}. 

In this paper, we present two innovative neural compression schemes which extend and surpass the work identified in \cite{deep_cpri}. These schemes are designed to efficiently reduce the bit rate over CPRI links while ensuring high-quality signal reconstruction at the receiving end.
Specifically, both proposed schemes employ a non-linear transformation, executed through a recurrent neural network (RNN), to process the input signal into a compact latent vector. 
The first scheme applies scalar or vector quantization to this latent vector on an element-wise or block-wise basis, respectively.
The second scheme, on the other hand, rounds each element of the latent vector to the nearest integer value, and utilizes a specially designed probability prediction network to estimate the likelihood of each quantized element's occurrence. Then, an arithmetic encoder translates these probabilities into a compact bit sequence for efficient CPRI link transmission.

In addressing the practical deployment challenges of our proposed neural compression scheme, we focus on two critical scenarios: 1) the presence of multiple CPRI links between the RRH and BBU, and 2) the existence of storage limitations on these devices. Each scenario presents unique requirements and opportunities for optimization.
\begin{enumerate}[leftmargin=0.45cm]
    \item In the first scenario, the CPRI links exhibit varied levels of reliability (as indicated by their packet error rates) and successive refinement transmission is a natural solution -- the bit sequence for a coarse reconstruction is transmitted over the most reliable CPRI link, while supplementary refinement data is scheduled over less reliable links. To effectively implement this strategy, we have developed a specialized neural network architecture tailored for successive refinement. This design not only accommodates the diverse reliability of CPRI links but also ensures high-quality signal reconstruction.
    \item In the second scenario, where storage capacity on the devices is constrained, we propose a novel solution through the development of a unified compression model. This model is adept at navigating different rate-distortion (R-D) trade-offs, a capability previously explored in variable-rate image compression research \cite{modulated_ae, cui2021asymmetric, qsfactor, deepjscclpp}. Our architecture draws inspiration from significant works in the field, yet it distinguishes itself by its efficiency, specifically tailored for CPRI compression. This efficiency is pivotal in achieving the desired storage reduction without compromising on the quality of compression, thereby addressing a crucial need within future wireless communication systems.
\end{enumerate}

% We then consider facilitating the deployment of the proposed neural compression scheme to the following use cases where 1). multiple CPRI links are available between the RRH and the BBU, 2). there is a storage constraint on these devices. For the first case, different CPRI links have different levels of reliability (reflected by their packet error rates) and successive refinement transmission is a natural solution where the bit sequence for a coarse reconstruction is transmitted over the most reliable link while the refinement information is scheduled on the less reliable ones. We delicately design a neural network architecture for the successive refinement model and achieves satisfactory reconstruction.
% We also consider to reduce the storage by introducing a unified compression model which is capable to achieve different R-D trade-offs. This has been investigated in the literature known as variable rate image compression \cite{modulated_ae, cui2021asymmetric, qsfactor, deepjscclpp}. We introduce a much efficient architecture to achieve the storage reduction objective for CPRI compression.

Our main contributions are summarized as follows:
\begin{itemize}[leftmargin=0.45cm]
    \item { We present two AI-native neural compression schemes that substantially advance traditional CPRI compression techniques, providing a principled framework for fronthaul bit-rate reduction in NextG Cloud RAN architectures.} By integrating non-linear transformation functions and employing unique latent quantization techniques, these schemes enhance the R-D performance, demonstrating a critical leap forward in compression efficiency.
    \item Through a statistical analysis of input signal distributions, we uncover the inherent correlation among signal elements, a revelation that challenges the adequacy of previous approaches and motivates our use of non-linear transformations. This insight directly informs the design of a successive refinement compression model tailored for environments with CPRI links of varying reliability, as well as a unified model that adeptly navigates storage constraints without sacrificing performance.
    \item Extensive numerical experiments validate the superior performance of our proposed schemes over traditional benchmarks, particularly in terms of error vector magnitude (EVM). Additionally, we demonstrate the robustness of our models under diverse conditions, including mismatched signal-to-noise ratios (SNRs), modulation formats, and filter taps. {The practical efficacy of both the unified model for storage optimization and the successive refinement strategy for leveraging link reliability further highlights the viability of our design for robust  6G fronthaul standardization.}
\end{itemize}

% \begin{itemize}
%     \item We propose two novel neural compression schemes where both of them inherit the decimation and block scaling modules from the conventional CPRI compression literature yet they employ non-linear transformation functions followed by different  latent quantization algorithms for improved R-D performance.
%     \item Solid analysis are provided on the distribution of the input to the CPRI compression module. We reveal that the elements of the input signal are correlated which indicates the previous solutions in the literature are sub-optimal motivating the introduction of the non-linear transformation functions.
%     \item A successive refinement compression model is designed when multiple CPRI links with different level of reliability are available between the RRH and the BBU. Novel neural network architecture and the corresponding training methodology are developed which are shown to be effective.
%     \item A unified model is proposed which is capable to achieve different R-D performance using a single set of non-linear transformation functions. This mitigates the storage constraints at both sides, especially for the RRH.
%     \item Numerical experiments are carried out to verify the effectiveness of the proposed schemes over the conventional benchmarks in terms of EVM. We also evaluate the robustness of the trained neural compression module under mismatched SNR, modulation formats and number of taps. Furthermore, the proposed unified model and the successive refinement model work well for their specific scenarios.
% \end{itemize}

{\it Notations}: Throughout the paper, scalars are represented by normal-face letters (e.g., $x$),  while uppercase letters (e.g., $X$) represent random variables. Matrices and vectors are denoted by bold {upper} and {lower} case letters (e.g., $\bm{X}$ and $\bm{x}$), respectively. Transpose and Hermitian operators are denoted by $(\cdot)^\top$, $(\cdot)^\dagger$, respectively. $\Re(x)$ ($\Im(x)$) denotes the real (imaginary) part of {a complex variable} $x$. $\lceil x \rceil$ ($\lfloor x \rfloor$) denote the nearest integer which is larger (smaller) than $x$. We denote by $|\bm{x}|$ and $L_{\bm{x}}$ the magnitude and length of a vector $\bm{x}$, respectively. We will also use the notation $[N]$ to denote the set of positive integers up to $N$, i.e., $[N] = \{1, 2, \ldots, N\}$.

\begin{figure*}
     \centering
     \begin{subfigure}{1.2\columnwidth}
         \centering
         \includegraphics[width=\columnwidth]{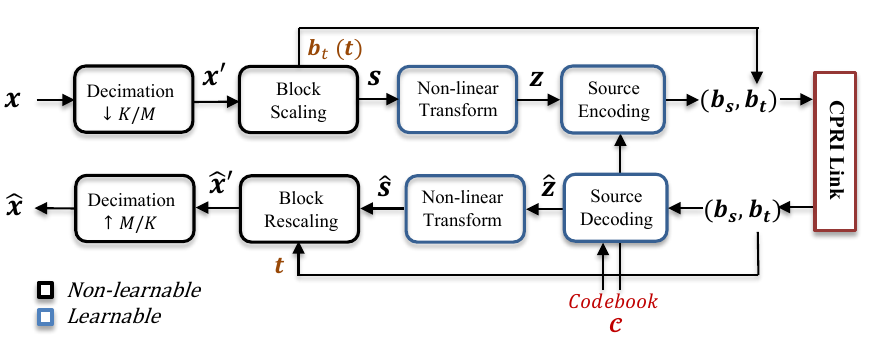}
         \caption{}
     \end{subfigure}
     \begin{subfigure}{0.34\columnwidth}
         \centering
         \includegraphics[width=\columnwidth]{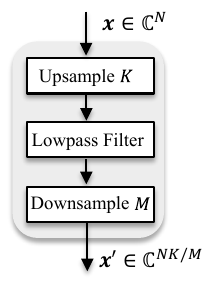}
         \caption{}
     \end{subfigure}

  \caption{(a) The flowchart of the unified CPRI compression procedure. In particular, the non-linear transformation only applies to the proposed neural compression algorithms and is parameterized by neural networks. The codebook, $\mathcal{C}$, can be obtained either by clustering (i.e., in the conventional schemes) or learning (i.e., the proposed latent quantization and neural compression schemes).
  (b) Details of the decimation module, where the input is a $N$-dimensional vector.}
\label{fig:flowchart}
\end{figure*}

\section{Problem Formulation}\label{sec:system_model}

\subsection{System model}
This section presents the model for signal compression over the CPRI link within the C-RAN architecture, addressing both uplink and downlink scenarios, as shown in Fig. \ref{fig:fig_system}.
The process begins with the generation of an information bit sequence. This bit sequence undergoes channel coding and modulation to produce Quadrature Amplitude Modulation (QAM) symbols, denoted by $\bm{c} \in \mathbb{C}^{N_{sym}}$, where $N_{sym}$ denotes the number of QAM symbols.

In the uplink scenario, mobile users utilize Single Carrier-Frequency Division Multiple Access (SC-FDMA): the QAM symbols, $\bm{c}$, are first converted to the frequency domain by an $N_{sym}$-point discrete Fourier transform (DFT). The resulting signal vector is then zero-padded to expand it into an $N_{fft}$-dimensional vector, which is subsequently transformed back to the time domain via an $N_{fft}$-point inverse DFT (IDFT).
A Cyclic Prefix (CP) consisting of $N_{cp}$ elements is added to this signal to combat the effects of multipath propagation and maintain subcarrier orthogonality across the transmission link.
After passing through the uplink channel, the signal received at the RRH, denoted by $\bm{x}\in \mathbb{C}^{N_{fft} + N_{cp}}$, will be forwarded to the CPRI compression module for efficient transmission over the CPRI link. 

In the downlink scenario, the BBU directly assigns the QAM symbols $\bm{c}$ to the OFDMA subcarriers, leaving $N_{\text{fft}} - N_{\text{sym}}$ unoccupied subcarriers as a guard band. Then, the QAM symbols are transformed to the time domain using an $N_{\text{fft}}$-point IDFT. 
Typically, the BBU appends a CP to these symbols before feeding them into the CPRI compression module for transmission to the RRH.
However, given the absence of channel noise in the downlink scenario, the CP merely replicates the final $N_{cp}$ elements of an OFDM symbol. Leveraging this redundancy, we can improve the compression efficiency by removing the CP prior to compression. This modification yields a CP-less signal, represented as $\bm{x} \in \mathbb{C}^{N_{fft}}$, that undergoes compression within the CPRI compression module. The CP will be re-integrated at the RRH, ensuring the signal's integrity during transmission while enhancing the overall compression efficiency over the CPRI link.

\subsection{CPRI compression}\label{sec:IIB}
The CPRI compression module at the RRH (BBU) for uplink (downlink) compression consists of four principal modules: decimation, block scaling, non-linear transformation, and source encoding.
The comprehensive workflow of the CPRI compression process is illustrated in Fig.~\ref{fig:flowchart}. In the following, we delve into the specifics of each module.

\subsubsection{Decimation}\label{sec:decimation}
The decimation module is introduced to reduce the time domain redundancy caused by the guard band which carries no information.
% \footnote{It is also possible due to the sampling rate of the analogy digital converter (ADC) exceeds that of the Nyquist sampling rate. In this paper, we assume the oversampling of ADC has been properly handled for simplicity.}
To this end, standard multi-rate filter is adopted. We first upsample the input signal, $\bm{x}$, by a factor of $K$. This is achieved by inserting zeros among samples:
\begin{equation}
    \bm{x}^{u}(j) = \left\{
    \begin{aligned}
    \bm{x}(n)  & , & j = Kn; \\
    0& , & j \neq Kn,
    \end{aligned}
    \right.
    \label{eq:upsample}
\end{equation}
and passing $\bm{x}^{u}$ through a low-pass filter to preserve the spectrum of interest while discarding the others, as zero insertion in the time domain leads to a periodic replication in the frequency domain.
Then, the filtered signal is downsampled by $M$ times, where $M>K$ such that the number of samples is reduced after the upsampling and downsampling processes. The output of the Decimation module, denoted by $\bm{x}^\prime\in\mathbb{C}^{N^\prime}$, where $N^\prime=(N_{fft}+N_{cp})K/M$ and $N^\prime=N_{fft}K/M$ for the uplink and downlink scenarios, respectively.

\subsubsection{Block Scaling}
The block scaling module scales the input signal $\bm{x}^\prime$ to limit the maximum value. 
To this end, $\bm{x}^\prime$ is first partitioned into multiple blocks where each block is comprised of $N_s$ complex symbols. For the $k$-th block, we first determine the scaling factor, $t_k$, as follows:
\begin{align}
    t^\prime_k = \max_{i\in [k N_s, (k+1)N_s)} \{|\Re(\bm{x}^\prime(i))|, |\Im(\bm{x}^\prime(i))|\},\notag
\end{align}
\begin{equation}
    t_k = \left\{
    \begin{aligned}
    \lceil t^\prime_k \rceil  & , & \lceil t^\prime_k \rceil \leq 2^{Q_s}-1, \\
    2^{Q_s}-1& , & \lceil t^\prime_k \rceil > 2^{Q_s}-1.
    \end{aligned}
    \right.
    \label{eq:scale_t}
\end{equation}
Note that $t_k$ is a positive integer, and can be represented using $Q_s$ bits. The corresponding scaled output, $\bm{s} = [\bm{s}_1^\top, \cdots, \bm{s}_{N_t}^\top]^\top$, is given by
\begin{align}
    \bm{s}_k = \bm{x}_k^\prime/t_k,
    \label{eq:scaled_s}
\end{align}
where $\bm{x}_k^\prime$ is the $k$-th block of $\bm{x}^\prime$ and $N_t \triangleq \lceil \frac{N^\prime}{N_s} \rceil$.
%When $Q_s$ is large, the maximum value of $\bm{s}$ is $1$.

\subsubsection{Transformation and source encoding}
We next transform the scaled data $\bm{s}$ into a bit sequence. Existing methods in CPRI compression adopt scalar quantization or vector quantization approaches:
\begin{itemize}
    \item Scalar quantization {\cite{scalar_cpri}}. In this approach, each element of the input signal $\bm{s}$ is quantized individually. A one-dimensional quantization codebook, $\mathcal{C} \in \mathbb{R}^{2^Q}$ with $2^Q$ levels, is trained via an iterative algorithm \cite{quantization_book}. The I/Q components of the quantized symbols $\hat{\bm{s}}(i)$ are given by 
    \begin{align}
    \Re(\hat{{s}}(i)) = \argmin_n |\mathcal{C}_n - \Re({{s}}(i))|, \notag \\
    \Im(\hat{{s}}(i)) = \argmin_n |\mathcal{C}_n - \Im({{s}}(i))|.
    \label{eq:scalar_quantize}
    \end{align}
    Each I/Q component is then represented using $Q$ bits. 
    \item Vector quantization \cite{vq_cpri}. Instead of quantizing $\bm{s}$ element-wisely,
    vector quantization organizes the input signal $\bm{s}$ into blocks, where each block consists $b$ I/Q components. A specialized codebook $\mathcal{C} \in \mathbb{R}^{2^{bQ}\times b}$, where each vector component is quantized using $Q$ bits, is generated using a modified Lloyd algorithm \cite{vq_cpri}. This method not only retains the advantages of scalar quantization but also reduces redundancy significantly by considering the statistical dependencies among multiple components. To further improve the compression efficiency, this method assigns different number of bits to different entries of the codebook, as opposed to uniform bit assignment, where each entry of the codebook is represented by $bQ$ bits. This can be achieved by variable-length source coding algorithms, taking the frequency of entries into account.
\end{itemize}

The output of the transformation and source encoding modules is a bit sequence, $\bm{b_s}$, representing the quantized version of $\bm{s}$. 
This sequence is combined with $\bm{b_t}$, the bit representation of the scaling vector $\bm{t} \triangleq [t_1, \ldots, t_{N_t}]$ introduced in \eqref{eq:scale_t}, to form the complete bit sequence $\bm{b}$, which is then transmitted over the CPRI link. The compression ratio (CR) for this process is calculated as follows
\begin{equation}
\text{CR} = \frac{ L_{\bm{b}}}{2\times 15 \times N^\prime} = \frac{K(2\alpha Q N^\prime + Q_sN_t)/M}{30N^\prime}.
    \label{eq:eq_cr}  
\end{equation}
Here, the denominator represents the total number of bits required to encode the original signal without compression, where $2$ accounts for the I and Q components, and $15$ refers to the bit precision per I/Q component by convention \cite{scalar_cpri}. The numerator represents the total bits in the compressed signal $\bm{b}$, where
\begin{itemize}
    \item $\alpha \in (0, 1]$ denotes the reduction of bit cost from the source coding algorithm. 
    \item $Q$ is the number of bits per quantized I/Q component.
    \item $Q_s$ is the number of bits used to represent the scaling factor.
    \item $N_t$ is the total number of scaling factors applied.
    \item $N^\prime$ is the sequence length, adjusted for both uplink and downlink scenarios.
\end{itemize}

\subsubsection{Decoder}
The CPRI decompression process is a direct mirror of the compression sequence. Upon reception of the bit sequence $\bm{b}$ from the CPRI link, the decoder segments it into $\bm{b}_t$ and $\bm{b}_s$, which are used to reconstruct the scaling vector $\bm{t}$ and the quantized signal $\hat{\bm{s}}$, respectively. The block rescaling operation is conducted as follows
\begin{align}
    \hat{\bm{x}}^\prime_k &= \hat{\bm{s}}_k {t}_k, \notag \\
    \hat{\bm{x}}^\prime = [&\hat{\bm{x}}^{\prime \top}_1, \cdots, \hat{\bm{x}}^{\prime \top}_{N_t}]^\top.
\end{align}

Finally, we adopt a multi-rate filtering process where $\hat{\bm{x}}^\prime$ is first upsampled by a factor of $M$, filtered to retain only the necessary spectral components, and then downsampled by $K$, resulting in the final time-domain reconstruction, $\bm{\hat{x}}$. Note that in the downlink scenario, a CP is added to $\bm{\hat{x}}$ before transmission through the channel. In the uplink scenario, synchronization is necessary to locate the start of the SC-FDMA frame, followed by the removal of CP from $\bm{\hat{x}}$ to prepare for further processing.

As highlighted in \cite{vq_cpri}, instead of assessing the performance of CPRI compression algorithms solely in the time domain, it is crucial to evaluate the quality of signal reconstruction in the frequency domain, where the actual data transmission occurs. 
The frequency-domain representations of the original and reconstructed signals, denoted by $\bm{x}_f$ and $\hat{\bm{x}}_f$, respectively, are obtained by performing a DFT on their time-domain counterparts. The error vector magnitude (EVM) is defined as
\begin{equation}
    \text{EVM} (\%) = \sqrt{\frac{\sum_{n\in \mathbb{O}}|\hat{\bm{x}}_f(n) - {\bm{x}}_f(n)|^2}{\sum_{n\in \mathbb{O}} |{\bm{x}}_f(n)|^2}}\times 100\%,
    \label{eq:evm_fd}
\end{equation}
where $\mathbb{O}$ denotes the set of occupied subcarrier indices, with a cardinality equal to $N_{sym}$. The EVM can also be expressed in decibels (dB):
\begin{equation}
    \text{EVM} \, (\text{dB}) = 10 \log_{10}\Big(\frac{\sum_{n\in \mathbb{O}} |{\bm{x}}_f(n)|^2}{\sum_{n\in \mathbb{O}}|\hat{\bm{x}}_f(n) - {\bm{x}}_f(n)|^2}\Big).
    \label{eq:evm_fd_dB}
\end{equation}

% Finally, the processing at the RRH and BBU for both uplink and downlink scenarios are shown in Fig. \ref{fig:fig_system} (b) and (c), respectively.

\begin{figure}[!t]
\centering
\includegraphics[width=0.7\columnwidth]{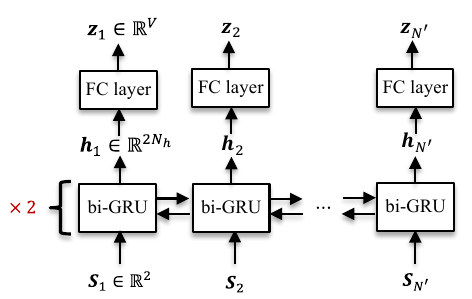}
\caption{The neural network architecture of the non-linear transformation function, $f_s(\cdot)$ at the encoder which transforms input signal $\bm{S}$ to the latent vector, $\bm{z}$.}
\label{fig:fig_NN}
\end{figure}

\section{Neural CPRI Compression with Latent Quantization}\label{sec:latent_quantize}
\subsection{Motivation}
The non-linear transformation and source encoding modules form the cornerstones of CPRI compression. As detailed in Section \ref{sec:IIB}, traditional solutions use a scalar or vector quantization scheme. However, both schemes rely on a unified codebook, wherein all the elements or blocks of elements from the input signal $\bm{s}$ are quantized in accordance with predefined levels. These schemes function by directly applying quantization to block-scaled data $\bm{s}$. Such methods achieve optimal results under the assumption that the elements or blocks being quantized are independent and identically distributed (i.i.d.). However, this assumption is invalid by analyzing the output of the decimation module, as demonstrated in Proposition \ref{prop:1}.

\begin{prop}\label{prop:1}
The output of the decimation module, $\bm{x}^\prime$, exhibits zero mean, i.e., $\mathbb{E}({x}^\prime_n) = 0$. The covariance for the elements of $\bm{x}^\prime$  is given by:
\begin{align}
\text{Cov}({x}^\prime_n, {x}^\prime_m) = \frac{P}{{N^\prime}} \sum_{k=k^*}^{N^\prime-k^*-1} e^{j\beta(n-m)k},
\end{align}
%\underbrace{\cos(\beta (n-m) k)}_{\text{real part}} \notag \\ &+ j \underbrace{\sin(\beta (n-m) k)}_{\text{imag part}}
$\forall n, m \in [N^\prime], \beta \triangleq \frac{2\pi}{N^\prime}$, where $k^* > 0$ represents the number of OFDM subcarriers dedicated to the guard band, and $P$ is the average power of the QAM symbols.
\end{prop}

\begin{proof}
See Appendix \ref{sec:APPA}.
\end{proof}

Proposition \ref{prop:1} reveals that the individual elements (or blocks) generated by the decimation module do not adhere to the i.i.d. assumption, nor does its block-scaled version, $\bm{s}$.
Therefore, the scalar and vector quantization approaches in \cite{scalar_cpri} and \cite{vq_cpri} are sub-optimal. 
To address this suboptimality, our idea is to employ a learnable non-linear transformation, detailed in Section \ref{sec:IIIB}, to transform the scaled data $\bm{s}$ into a latent space where the elements are more amenable to quantization.
Following the transformation, Section \ref{sec:III_quan} discusses the specific techniques used to quantize this neural transformed latent data. Different quantization strategies are explored to optimize the compression of the transformed data while maintaining the integrity and quality of the reconstructed signals.

Building on the above foundation, Section \ref{sec:neural_compress} advances the discussion by detailing how the compression model employs a probabilistic learning approach. This approach involves predicting the probability distribution of the latent data, enabling more efficient and adaptive compression.

\subsection{Learning-based Non-linear Transformations}\label{sec:IIIB}
Considering that $\bm{s}$ is a time-domain sequence, we utilize recurrent neural networks (RNNs) to parameterize non-linear transformation functions. Specifically, we employ bi-directional Gated Recurrent Units (bi-GRUs) \cite{GRU}, followed by fully connected layers, to facilitate this transformation.

To illustrate this process using the CPRI compression scenario for uplink (which is similarly applicable to downlink), let us first consider the input signal, $\bm{s} \in \mathbb{C}^{N^\prime}$. This signal is initially transformed into a real-valued matrix, $\bm{S} \in \mathbb{R}^{2 \times N^\prime}$, represented as
\begin{equation}
    \bm{S} = \begin{bmatrix} \Re(\bm{s}^\top) \\ \Im(\bm{s}^\top) \end{bmatrix}.
    \label{eq:big_S}
\end{equation}
Subsequently, $\bm{S}$ undergoes a non-linear transformation through the function $f_s(\cdot)$, generating a latent representation, $\bm{z} \in \mathbb{R}^{V N^\prime}$: 
\begin{equation}
    \bm{z} = f_s(\bm{S}).
    \label{eq:latent_z}
\end{equation}

The neural network architecture for the transformation function $f_s(\cdot)$ is depicted in Fig. \ref{fig:fig_NN}. This structure incorporates a two-layer bi-GRU with each layer having a hidden dimension of $N_h$. At each time step $j$, the input to the bi-GRU module, denoted by $\bm{S}_j$, is a two-dimensional vector corresponding to the $j$-th column of $\bm{S}$.

The output from the bi-GRU at the $j$-th step, $\bm{h}_j$, is a vector of length $2N_h$, which results from concatenating the outputs of the forward and backward passes, each contributing a vector of length $N_h$. This concatenated vector $\bm{h}_j$ is then processed by a fully connected (FC) layer, which transforms it into the latent vector $\bm{z}_j$ of dimension $V$. The final latent representation $\bm{z}$ for the entire sequence is assembled by concatenating each $\bm{z}_j$ across all time steps:
\begin{equation}
    \bm{z} = [\bm{z}_1^\top, \ldots, \bm{z}_{N^\prime}^\top]^\top.
    \label{eq:eq_latent_z}
\end{equation}

Given the latent $\bm{z}$, we quantize it (detailed later in Section \ref{sec:III_quan}) and obtain the quantized form $\hat{\bm{z}}$.
The corresponding bit sequence is then transmitted to the destination over the CPRI link. Upon receipt, the destination recovers $\hat{\bm{z}}$ and feeds it into the latent decoding function, $g_d(\cdot)$, which is designed to reconstruct $\bm{S}$. Note that the neural network architecture for $g_d(\cdot)$ mirrors that of $f_s(\cdot)$ with the exception of the input dimension, which is $V$ for $g_d(\cdot)$, as opposed to two.

During the decoding process, for each $j$-th time slot, $\bm{\hat{z}}_j$ is fed into the decoder. The bi-GRU then generates a hidden vector $\bm{\hat{h}}_j$ of size $2N_h$. This vector is subsequently processed by a FC layer that outputs a length-2 vector $\bm{\hat{s}}_j$. The complete reconstructed real-valued matrix, $\bm{\hat{S}}$, is compiled by stacking these vectors sequentially
\begin{equation}
    \bm{\hat{S}} = [\bm{\hat{s}}_1, \ldots, \bm{\hat{s}}_{N^\prime}].
\end{equation}
The complex reconstructed signal, denoted by $\hat{\bm{s}}$, is then formed by taking the first row of $\bm{\hat{S}}$ as the real component and the second row as the imaginary component.

\subsection{Quantization}\label{sec:III_quan}
In this section, we detail the process of quantizing the latent vector $\bm{z}$ to derive its quantized counterpart $\hat{\bm{z}}$.

\subsubsection{Uniform Quantization}
We first consider uniform quantization of the latent vector. We start by clipping each element of the output from the latent encoding function within the range $[0, 1]$:
\begin{equation}
    \bm{z} = \text{clip}(f_s(\bm{S}), 0, 1).
\end{equation}

We then assume that each component of $\bm{z}$ is i.i.d. and follows a uniform distribution, $\mathcal{U}(0, 1)$. This assumption allows us to apply uniform scalar quantization individually to each element, which can be written as
\begin{equation}
    \hat{{z}}_i = \frac{\lfloor(2^Q-1) {z}_i + \frac{1}{2}\rfloor}{2^Q-1},
\end{equation}
where $Q$ represents the number of bits per element, and $\lfloor x + \frac{1}{2}\rfloor$ is the rounding operation. Notably, this rounding operation is not differentiable, which complicates the application of gradient descent during training phases.

To overcome this, we employ the straight through estimator (STE) \cite{deepjsccq}, which allows the gradients of $\hat{\bm{z}}$ to pass through to ${\bm{z}}$, facilitating the update of neural network parameters within the latent encoding function $f_s(\cdot)$:
\begin{equation}
    \nabla_{{\bm{z}}}\mathcal{L} \leftarrow \nabla_{\hat{\bm{z}}}\mathcal{L},
    \label{eq:eq_ste}
\end{equation}
where $\mathcal{L}$ denotes the loss function. As detailed in \eqref{eq:evm_fd}, our focus lies on optimizing the signal reconstruction specifically in the occupied subcarriers. Therefore, the loss function is defined as
\begin{align}
    \mathcal{L}_{uq} = \frac{\|\bm{x}_f - \bm{\hat{x}}_f\|_2^2}{\|\bm{x}_f\|^2_2}.
    \label{eq:uq_loss}
\end{align}

The total number of bits required by the uniform quantization scheme can be written as
\begin{equation}
    L_{\bm{b}} = VQN^\prime + Q_s \left\lceil \frac{N^\prime}{N_s} \right\rceil.
\end{equation}
However, a notable limitation of the uniform quantization scheme is the assumption that each element of $\bm{z}$ adheres to an i.i.d. uniform distribution. This assumption can lead to sub-optimal quantization performance, particularly when the actual distribution of the latent vector $\bm{z}$ diverges from this assumed uniformity. The implications of this divergence will be further explored and substantiated by the simulation results presented in Section \ref{sec:experiment}.

\subsubsection{Vector Quantization}
The strong i.i.d. uniform assumption can be relaxed by vector quantization: we partition $\bm{z}$ into blocks and treat each block as an independent unit that follows a categorical distribution, rather than assuming each element within the blocks is i.i.d. Specifically, we define $\bm{z} = [\bm{z}_1, \ldots, \bm{z}_{N_b}]$,\footnote{With a slight abuse of notation, here we denote by $\bm{z}_i$ the individual blocks after partition. This usage differs from \eqref{eq:eq_latent_z}, wherein $\bm{z}_i$ specifies the output of the GRU corresponding to the $i$-th time step.} where each $\bm{z}_i \in \mathbb{R}^{b}$ and $bN_b = VN^\prime$. Each block $\bm{z}_i$ is independently quantized using a learned codebook, a process that mirrors the training approach of the Vector Quantization Variational Autoencoder (VQVAE) \cite{vqvae}. It is crucial to note that our goal diverges from that of the VQVAE. While the VQVAE primarily aims to generate new data samples, our objective is to develop a discrete latent representation that ensures high-quality reconstruction performance.
Given a bit budget of $Q$ bits per element, this necessitates a codebook $\bm{E} \in \mathbb{R}^{2^{bQ} \times b}$, and we detail the training process in the following.

Rather than restricting the elements of $\bm{z}$ to the interval $[0, 1]$ as required by the uniform quantization method, we enhance vector quantization by selecting the nearest vector from codebook $\bm{E}$ for each block $\bm{z}_i$. The quantization process is given by
\begin{align}
    \hat{\bm{z}}_{i} &= \bm{E}_j, \notag \\
    s.t.~j &= \argmin_{k \in [2^{bQ}]} \|\bm{z}_i - \bm{E}_k\|^2_2.
\end{align}
where $\hat{\bm{z}}_i$ is the quantized output of the $i$-th block.

Similar to the uniform quantization, vector quantization does not inherently support gradient-based optimization due to its non-differentiable nature. To overcome this, we pass the gradient from the quantized latent $\hat{\bm{z}}$ back to the original latent $\bm{z}$. Optimizing the vector quantization scheme involves tuning both the non-linear transformation functions and the codebook, ensuring effective reconstruction performance. The loss function, integrating vector quantization, is structured as
\begin{align}
    \mathcal{L}_{vq} = \frac{\|\bm{x}_f - \bm{\hat{x}}_f\|_2^2}{\|\bm{x}_f\|^2_2}  + \|\text{sg}[\bm{z}] - \bm{\hat{z}}\|^2_2 + \beta \|\bm{z} - \text{sg}[\bm{\hat{z}}]\|^2_2,
    \label{eq:vqvae_loss}
\end{align}
where $\text{sg}[\cdot]$ denotes the stop gradient operator. The first term of the loss function optimizes the non-linear transformation functions $f_s(\cdot)$ and $g_d(\cdot)$. The second term is designed to learn a good codebook by gradually fitting it to the encoder output, $\bm{z}$. Since the second term aims to optimize the codebook rather than the latent, we apply the stop gradient operator, $\text{sg}[\cdot]$ to $\bm{z}$. 
The final term is known as the commitment loss, which forces the latent to converge to the learned codebook, $\bm{E}$. In this paper, we set $\beta = 1$ as the reconstruction performance is not sensitive with respect to different $\beta$ values.

This vector quantization strategy reduces the limitations of assuming i.i.d. uniform distribution in latent quantization by allowing for a more nuanced grouping and coding of data blocks. This method can be viewed as an enhancement of traditional vector quantization approaches like those discussed in \cite{vq_cpri}, with the critical difference being that our codebook is derived through learning rather than predefined algorithms such as Lloyd's algorithm. However, it is essential to acknowledge that our assumption about the i.i.d. nature of blocks could still deviate from reality. To address potential discrepancies and further refine our model, we will explore advanced compression frameworks in Section \ref{sec:neural_compress}.

\begin{figure}[!t]
\centering
\includegraphics[width=0.6\columnwidth]{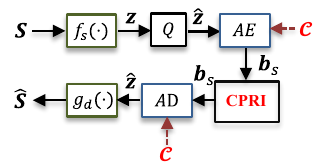}
\caption{The pipeline of the proposed neural compression model. The codebook, $\mathcal{C}$, is constructed using the probability prediction net illustrated in \eqref{eq:codebook}.}
\label{fig:fig_neural_compress}
\end{figure}

\section{Neural CPRI Compression with Probability Prediction} \label{sec:neural_compress}
In the latent quantization approach outlined in \eqref{eq:latent_z} and \eqref{eq:eq_latent_z}, we generate $N^\prime$ vectors $\{\bm{z}_j\}_{j\in [ N^\prime]}, \bm{z}_j \in \mathbb{R}^V$, through a non-linear transformation function $f_s(\cdot)$. These vectors can be assembled into a matrix $\bm{Z}\in \mathbb{R}^{V\times N^\prime}$:
\begin{equation}
    \bm{Z} = [\bm{z}_1, \ldots, \bm{z}_{N^\prime}].
    \label{eq:eq_latent_Z}
\end{equation}
From an information-theoretical perspective, if the distribution of $\bm{Z}$ is known, it is possible to compress the latent vectors losslessly to a rate close to the entropy of $\bm{Z}$, denoted by $H(\bm{Z})$, using lossless source coding algorithms. However, in practical scenarios where the modulation order, channel responses, and noise levels are unknown, the precise distribution of both $\bm{s}$ and $\bm{Z}$ remains elusive. To address this, we propose using a neural network, which we refer to as the probability prediction net, parameterized by $\bm{\phi}$, to model the distribution of $\bm{Z}$. This network aims to approximate the fully factorized probability of the latent variables, thus facilitating a more streamlined and efficient compression process.

\subsection{Pipeline}\label{sec:pipeline}
We first overview the pipeline of the proposed neural compression approach with probability prediction introduced in the image compression literature \cite{balle2016end, balle2018variational}. 
Initially, each element of the matrix $\bm{Z}$ is rounded to the nearest integer, resulting in $\hat{\bm{Z}} = \lfloor \bm{Z}+\frac{1}{2}\rfloor$.
It is important to note that any random variable $Z_{i,j}$ within the range $(\hat{{z}}_{i,j} - \frac{1}{2}, \hat{{z}}_{i,j} + \frac{1}{2}]$ will be rounded to the same value $\hat{{z}}_{i,j}$.
Consequently, the probability of the quantized random variable $\hat{Z}_{i,j}$ can be written as
\begin{equation}
    \text{Pr}(\hat{Z}_{i,j} = \hat{{z}}_{i,j}) = P_{\phi_i}\left(\hat{{z}}_{i,j} + \frac{1}{2}\right) - P_{\phi_i}\left(\hat{{z}}_{i,j} - \frac{1}{2}\right),
    \label{eq:predicted_prob}
\end{equation}
where $P_{\phi_i}(\cdot), i\in [V]$ denotes the probability mass function (PMF) for the $i$-th channel (i.e., the $i$-th row of the latent $\bm{Z}$), and is parameterized by ${\phi}_i$. 
The overall learnable parameters $\bm{\phi}$ are a collection of the parameters of all $V$ channels, i.e., $\bm{\phi} \triangleq \{{\phi}_i\}_{i \in [V]}$.
The training methodologies for these parameters are elaborated upon later in Section \ref{sec:neural_compress_train}.

Subsequently, an arithmetic encoder processes $\hat{\bm{Z}}$ into a bit sequence $\bm{b}_s$ following the probability model described in \eqref{eq:predicted_prob}, which is then transmitted over the CPRI link. The length of $\bm{b}_s$ can be closely approximate the predicted entropy of $\hat{\bm{Z}}$ to within one bit, as given by
\begin{align}
    L_{\bm{b}_s} <  - \sum_{i,j} \log_2 \left[P_{\phi_i}\Big(\hat{{z}}_{i,j} + \frac{1}{2}\Big) - P_{\phi_i}\Big(\hat{{z}}_{i,j} - \frac{1}{2}\Big)\right] + 1.
    \label{eq:bit_cost}
\end{align}
Upon receiving the bit sequence, the arithmetic decoder is applied to reconstruct the latent matrix, $\hat{\bm{Z}}$, which is then input into the latent decoding function to produce the final reconstructed signal $\hat{\bm{s}}$. 

During the deployment phase, it is more practical to use a lookup table instead of dynamically generating probabilities with the probability prediction network for each quantized value $\hat{{z}}_{i,j}$. This approach leverages the fact that the range of possible values for $\hat{{z}}_{i,j}$ can be precomputed and stored efficiently. Formally speaking, for the $i$-th channel, denote the support set of $\hat{\bm{z}_i} = \{\hat{{z}}_{i,j}\}_{j\in [N^\prime]}$ as $[c_{i}^l, c_{i}^h] \cap \mathbb{Z}$, then the $j$-th entry of the lookup table corresponding to the $i$-th channel, denoted by $\mathcal{C}[i,j]$, is written as
\begin{align}
    \mathcal{C}[i,j] = P_{\phi_i}\Big(j + c_{i}^l + \frac{1}{2}\Big) - P_{\phi_i}\Big(j + c_{i}^l - \frac{1}{2}\Big),
    \label{eq:codebook}
\end{align}
satisfying $\sum_{j = 0}^{c_{i}^h - c_{i}^l} \mathcal{C}[i,j] = 1, \forall i$. 
We refer the readers to \cite{balle2018variational} for more details concerning the  probability prediction network which is omitted here due to the page limit.

\subsection{Training the Compression Module}\label{sec:neural_compress_train}
The rounding operation, which generates $\hat{\bm{Z}}$ from ${\bm{Z}}$, is not directly applicable in training due to its non-differentiable nature, as the gradient is essentially zero almost everywhere. To address this challenge during training, we follow \cite{balle2016end} and introduce a noise component sampled from a uniform distribution, $\mathcal{U}\left(-\frac{1}{2}, \frac{1}{2}\right)$, to approximate the rounding error, defining $\tilde{\bm{Z}} = {\bm{Z}} + \mathcal{U}(-\frac{1}{2}, \frac{1}{2})$. The probability density function (pdf) of the resulting distribution is given by
\begin{align}
    p(\tilde{\bm{Z}}|\bm{\phi}) = \prod_{i,j} \left(p_{\tilde{{Z}}_{i,j}|\bm{\phi}_i}(\bm{\phi}_i)*\mathcal{U}(-\frac{1}{2}, \frac{1}{2}) \right)(\tilde{{Z}}_{i,j}).
\label{equ:probv}
\end{align}

To optimize the neural compression framework, including the non-linear transformation functions $f_s(\cdot), g_d(\cdot)$ and the probability prediction net $P_{\bm{\phi}}(\cdot)$, we employ a loss function that targets both the quality of signal reconstruction and the efficiency of compression, reflecting a balanced R-D trade-off. The loss function is defined as
\begin{equation}
    \mathcal{L} = \lambda \frac{\|\bm{x}_f - \bm{\hat{x}}_f\|_2^2}{\|\bm{x}_f\|^2_2} + I,
    \label{eq:compress_loss}
\end{equation}
where $\lambda$ is a tuning parameter that adjusts the emphasis between reconstruction fidelity and compression rate, and $I$ represents the average bit cost per element, calculated as
\begin{align}
    I = \frac{1}{VN^\prime} \mathbb{E}_{{\bm{S}} \sim p(\bm{S})}\mathbb{E}_{\tilde{\bm{Z}} \sim q(\tilde{\bm{Z}}|\bm{S})} \left(- \log_2(p(\tilde{\bm{Z}}|\bm{\phi}))\right),
\label{equ:Iz}
\end{align}
where the posterior probability $q(\tilde{\bm{Z}}|\bm{S})$ is given by
\begin{align}
    q(\tilde{\bm{Z}}|\bm{{S}}) &=  \prod_i\mathcal{U}\Big(\tilde{{z}}_{i,j}|{z}_{i,j}-\frac{1}{2}, {z}_{i,j}+\frac{1}{2}\Big), \notag \\
    \bm{Z} &= f_s(\bm{S}).
\label{equ:post}
\end{align}
The posterior probability follows because $\tilde{\bm{Z}}$ is constructed by adding uniform noise to ${\bm{Z}}$. 

We point out that \eqref{equ:Iz} is essentially the cross entropy between the actual distribution of $\tilde{\bm{Z}}$, i.e., {$\int q(\tilde{\bm{Z}}|\bm{{S}}) p(\bm{S}) dS$} and the predicted distribution $p(\tilde{\bm{Z}}|\bm{\phi})$. Our goal is to minimize it by learning a $p(\tilde{\bm{Z}}|\bm{\phi})$ that closely matches the true distribution.
Various models, each characterized by different $\lambda$ settings, are trained to explore distinct points on the R-D curve, ensuring that the neural compression algorithm optimally balances the average compression rate and reconstruction quality in practical applications, as demonstrated in the simulations.

\begin{figure*}[t]
\centering
\includegraphics[width=0.8\linewidth]{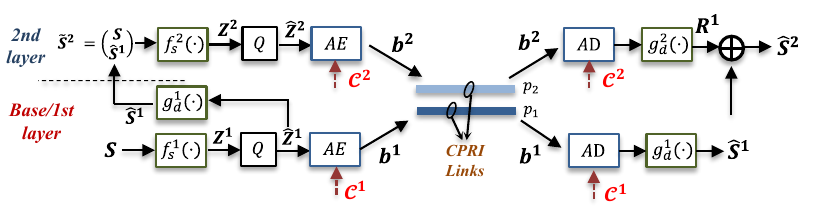}\\
\caption{An illustration of the successive refinement framework with $L = 2$. The encoder of the refinement (2nd) layer aims to produce a bit sequence, $\bm{b}_1$ to generate the residual, $\bm{R}^1$. The final reconstruction of the $L=2$ case is simply a summation of the coarse reconstruction $\hat{\bm{S}}^1$ and the residual $\bm{R}^1$.}
\label{fig:successive_refine}
\end{figure*}

\subsection{Successive Refinement Transmission}\label{sec:successive_refine}
In practical network architectures, there can be multiple (virtual) links connecting the RRH to BBU. Each of these links can exhibit distinct Qualities of Service (QoS), typically differentiated by their error probabilities. In this context, we consider the successive refinement transmission strategy \cite{scalar_cpri}, which leverages this diversity in link quality to enhance data transmission robustness and efficiency.

In successive refinement transmission, the links are organized in order of increasing error probability, with the most reliable link designated as the base/first layer. This base layer is responsible for transmitting a coarse version of the original signal, $\bm{s}$, providing a foundational reconstruction. Subsequent layers, associated with links of lower reliability, are tasked with transmitting refinement packets that incrementally enhance the signal based on the previously reconstructed output.
The key merit of the successive refinement model is that, even if an error occurs in the $(\ell+1)$-th packet, we are still able to obtain a usable version of the signal using the information from the first $\ell$ packets. 

We consider $L$ CPRI links, designating one as the base layer and the remaining $L-1$ as refinement layers. 
Both decimation and block scaling modules are performed only once, and the input to the successive refinement encoder, $\bm{S}$, remains consistent as defined in \eqref{eq:big_S}. The operation of the proposed successive refinement model is depicted in Fig. \ref{fig:successive_refine} and is described as follows.
\subsubsection{Base layer}
The base layer encoder, denoted as $f_s^{1}(\cdot)$, processes this input to generate $\bm{Z}^1$, with each element then quantized to integer values to form $\hat{\bm{Z}}^1$. This quantized latent representation is subsequently encoded into a bit sequence, $\bm{b}^1$, which is transmitted over the CPRI link.
Upon receipt, the base layer decoder, $g_d^1(\cdot)$, decodes $\bm{b}^1$ to produce a coarse reconstruction of the original signal, $\bm{\hat{S}}^1 = g_d^1(\hat{\bm{Z}}^1)$.
The process for generating refinement information begins with $\ell = 2$ and can be extended to scenarios where $\ell > 2$.

\subsubsection{Refinement layers}
As shown in Fig. \ref{fig:successive_refine}, the successive refinement encoder also has access to the the base layer decoder, $g_d^1(\cdot)$, thus, it is able to generate $\bm{\hat{S}}^1$, which is identical to the base layer reconstruction at the decoder side. 
The second-layer encoder concatenates the coarse reconstruction with the original signal, $\bm{{S}}$, to form a larger matrix $\bm{\widetilde{S}}^2 \in \mathbb{R}^{4 \times N^\prime}$, which can be written as $\bm{\widetilde{S}}^2 = \begin{bmatrix} \bm{{S}} \\ \bm{\hat{S}}^1 \end{bmatrix}$.
Then, the second-layer encoder, $f_s^2(\cdot)$, is employed: $\bm{Z}^2 = f_s^2(\bm{\widetilde{S}}^2)$. Note that $f_s^2(\cdot)$ is identical to $f_s^1(\cdot)$ except for the input dimension. The output $\bm{Z}^2$ will be quantized and arithmetically encoded to generate the bit sequence, $\bm{b}^2$. At the decoder side, the bit sequence is arithmetically decoded and we use a second-layer decoder, $g_d^2(\cdot)$, to estimate the residual information, $\bm{R}^1$. The reconstruction at the second layer can thus be written as $\hat{\bm{S}}^2 = \hat{\bm{S}}^1 + \bm{R}^1$. 

For $\ell > 2$, we have
\begin{align}
    \bm{\widetilde{S}}^{\ell} &= \begin{bmatrix} \bm{{S}} \\ \bm{\hat{S}}^{\ell-1} \end{bmatrix}; \notag\\ 
    \bm{Z}^{\ell} &= f_s^{\ell}(\bm{\widetilde{S}}^{\ell});\notag \\ 
    \bm{R}^{\ell-1} &= g_d^{\ell}(\hat{\bm{Z}}^{\ell});\notag\\
    \bm{\hat{S}}^{\ell} &= \bm{\hat{S}}^{\ell-1} +  \bm{R}^{\ell-1}.
    \label{eqn:sr_ell}
\end{align}
After obtaining $\bm{\hat{S}}^{L}$, $\hat{\bm{x}}_f$ is generated following the same procedure illustrated in Section \ref{sec:system_model} and the reconstruction performance is evaluated according to \eqref{eq:evm_fd} and \eqref{eq:evm_fd_dB}.

\subsubsection{Training methods}
The training of the successive refinement layers follows a structured approach. Initially, the neural networks for the base layer are trained via the standard training approach illustrated in Section \ref{sec:neural_compress_train}, with hyper parameters $\lambda_1$ to balance the rate and the distortion terms. 
Then, the neural network parameters of the $\ell$-th ($\ell>1$) layer are successively trained by fixing the parameters of all previous layers $\ell^\prime<\ell$, and  focusing on optimizing the current layer with the loss function
\begin{equation}
    \mathcal{L}_\ell = \lambda_\ell \frac{\|\bm{x}_f - \bm{\hat{x}}_f^\ell\|_2^2}{\|\bm{x}_f\|_2^2} + I_\ell,
    \label{eq:ell_loss}
\end{equation}
where $I_\ell$ denotes the estimated bit cost for compressing the latent $\bm{Z}^{\ell}$ and $\bm{\hat{x}}_f^\ell$ is obtained from $\bm{\hat{S}}^{\ell}$ in \eqref{eqn:sr_ell}. 

\begin{rem}
The choice of $\lambda_\ell$ is crucial; as the reconstruction accuracy improves with each layer, $\lambda_\ell$ needs to be adjusted to maintain an effective balance in the diminishing $\frac{\|\bm{x}_f - \bm{\hat{x}}_f^\ell\|_2^2}{\|\bm{x}_f\|_2^2}$ term, ensuring optimal performance across different layers.
\end{rem}

\begin{figure}[!t]
\centering
\includegraphics[width=0.7\columnwidth]{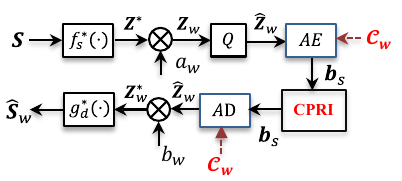}
\caption{Pipeline of the variable-rate compression framework, where we assume the model corresponding to the $w$-th compression level is applied.}
\label{fig:fig_var_rate}
\end{figure}

\subsection{Variable-Rate Compression}\label{sec:variable_rate}
Another critical concern of CPRI compression in practical deployments is managing the diverse reconstruction requirements across different users. This diversity necessitates maintaining multiple neural network configurations at both the RRHs and BBUs, which can lead to significant storage challenges, especially at the RRHs. To effectively address this issue, we introduce a variable-rate compression scheme that not only accommodates varied user demands but also significantly reduces the storage footprint. 
%In the section, we shall focus on the uplink CPRI compression to illustrate the scheme. 
The methodology is applicable to both uplink and downlink scenarios.

\subsubsection{Overview}
For each model with a specific $\lambda$ value, RRH has to store the corresponding weights for the bi-GRU and FC modules as well as the lookup table introduced in \eqref{eq:codebook} for entropy coding. This causes severe storage issue when the number of different $\lambda$ values is large.
Note that in our case, the lookup table $\mathcal{C}$ is relatively small. 
Thus, we seek to save the storage by introducing a shared non-linear transformation module, $f_s(\cdot)$ for all the models targeting different R-D performances.

To achieve this, we start from the model with the best reconstruction performance (trained under the largest $\lambda$ value), whose latent encoding and decoding functions are denoted as $f_s^*(\cdot)$ and $g_d^*(\cdot)$, respectively. For an input signal $\bm{S}$, we first obtain the latent vector, $\bm{Z}^* = f_s^*(\bm{S})$. Then, for the $w$-th ($w \in [W]$) reconstruction, we scale the latent by the scaling factor, $a_w \in (0, 1]$:
\begin{align}
    \bm{Z}_w &= a_w \bm{Z}^*, \\
    0 < a_1 < &a_2 < \cdots < a_W = 1. \notag
\end{align}
Note that for the special case when $w = W$, we have $a_W = 1, \bm{Z}_W = \bm{Z}^*$ which corresponds to the model with the best reconstruction performance. 
The $w$-th scaled latent is then quantized as:
\begin{align}
    \hat{\bm{Z}}_w &=  \left\lfloor  \bm{Z}_w + \frac{1}{2} \right\rfloor. \notag
\end{align}
Note that the scaling operation will change the underlying distribution of the quantized latent variable, thus, the probability prediction net, $P^*_{\bm{\phi}}(\cdot)$ optimized for the latent $\hat{\bm{Z}}^*$, is no longer feasible for the scaled one. To resolve this, for each $a_w, w \in [W-1]$, we train a new probability prediction net, denoted by $P^w_{\bm{\phi}_w}(\cdot)$, which will be detailed later. The decoding process also needs to be changed with respect to the scaling operation. To be precise, we introduce a set of re-scaling factors:
\begin{align}
    b_w = 1/a_w,  \quad w \in [W]
    \label{eq:bk}
\end{align}
for different models, and they will be directly multiplied with the quantized latents as:
\begin{align}
    \hat{\bm{Z}}^*_w &=  b_w \lfloor  a_w \bm{Z}^* + \frac{1}{2} \rfloor \approx \hat{\bm{Z}}^*. 
    \label{eq:quant_vr}
\end{align}
The approximation holds due to the fact that we choose $b_w = 1/a_w$. This is to ensure the input to the latent decoding function, $g_d^*(\cdot)$ has nearly identical distribution with its original input, $\hat{\bm{Z}}^*$, such that  $g_d^*(\cdot)$ can produce satisfactory reconstruction without modifying its parameters.

\begin{figure*}
     \centering
     \begin{subfigure}{0.66\columnwidth}
         \centering
         \includegraphics[width=\columnwidth]{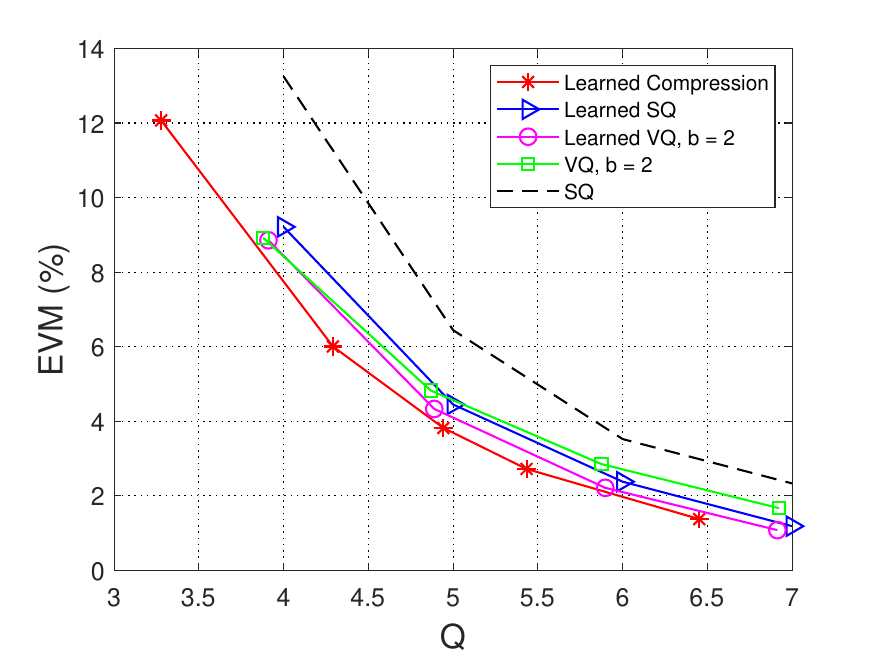}
         \caption{}
     \end{subfigure}
     \begin{subfigure}{0.66\columnwidth}
         \centering
         \includegraphics[width=\columnwidth]{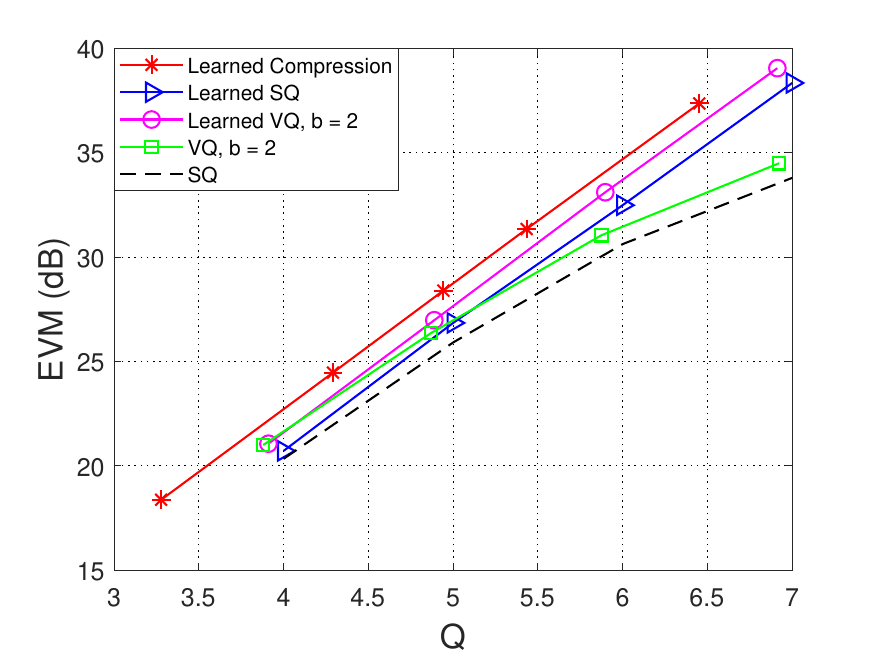}
         \caption{}
     \end{subfigure}
    %  \vspace{0.1cm}
     \begin{subfigure}{0.66\columnwidth}
         \centering
         \includegraphics[width=\columnwidth]{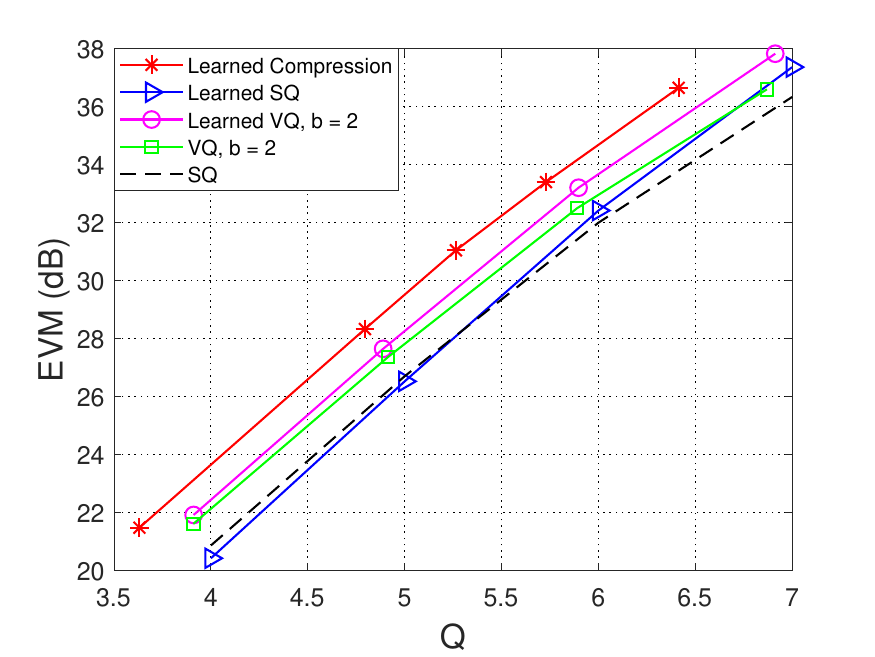}
         \caption{}
     \end{subfigure}
  \caption{Reconstruction performances of the proposed schemes for CPRI compression: (a) \& (b) the EVM performance (in percentage and dB) of the reconstructed signal in the downlink scenario; (c) EVM in dB for the uplink scenario where the cyclic prefix as well as the received signal are jointly compressed.}
\label{fig:final_simu}
\end{figure*}

\subsubsection{Obtaining $P^w_{\bm{\phi}_w}(\cdot)$ and $\mathcal{C}_w$}
We then illustrate the acquisition of new probability prediction net $P^w_{\bm{\phi}_w}(\cdot)$ and the codebook $\mathcal{C}_w$ corresponding to the $w$-th reconstruction. 
In particular, we keep the neural network parameters of the non-linear transformation functions, $f_s^*(\cdot)$ and $g_d^*(\cdot)$ to be fixed during training and randomly initialize the parameters of $P^w_{\bm{\phi}_w}(\cdot)$. The training objective can be expressed as:
\begin{equation}
    \min_{\bm{\phi}_w}  I_w(\bm{\phi}_w),
    \label{eq:min_phi}
\end{equation}
where $I_w$ follows the same definition in \eqref{equ:Iz} and is a function of the parameters, $\bm{\phi}_w$.
Note that we omit the reconstruction term in \eqref{eq:compress_loss} as it is independent with respect to $\bm{\phi}_w$ and the loss function $\mathcal{L}_w$ is simply $I_w$. In other words, we only need to train a probability prediction net which predicts the distribution of the scaled latent, $\hat{\bm{Z}}^*_w$, to produce an accurate estimation of $I_w$. Since the parameters of the latent encoding and decoding functions are fixed, the training process can be made efficient. After obtaining $P^w_{\bm{\phi}_w}(\cdot)$, we can construct the lookup table, $\mathcal{C}_w$ for each of the models with $w \in [W-1]$ following \eqref{eq:codebook}.

\section{Numerical Experiments}

\label{sec:experiment}
\subsection{Parameter Settings and Training Details}
To verify the efficacy of the proposed schemes, we consider a 5G scenario \cite{5G_setting}, where the channel bandwidth, subcarrier spacing, FFT size, and the sampling frequency are set to $20$ MHz, $60$ kHz, $512$ and $30.72$ MHz, respectively. The decimation module for both uplink and downlink scenarios first upsamples the signal by $5$ times followed by a low-pass filter, and finally downsamples the filtered signal by $8$ times. 

Next, we introduce the DNN parameters for the proposed latent quantization and neural compression schemes. The bi-GRU module of the non-linear transformation function of both schemes has 2 hidden layers with the hidden size equals to $N_h = 32$. We set $V = 2$ for the latent vectors $\{\bm{z}_j\}_{j\in [N^\prime]}; \bm{z}_j \in \mathbb{R}^{V}$. The batch size for both training and evaluation phases is set to $32$. Adam optimizer is adopted with a varying learning rate, initialized to $10^{-4}$ and dropped by a factor of $0.8$ if the validation loss does not improve in $20$ consecutive epochs.

\subsection{General Performance}

\subsubsection{Downlink case}
We first evaluate the R-D performance of the downlink signal. In this case, CP is removed and the input to the non-linear transformation function is of $640$-dimension. The modulation order is assumed to be $64$.  For both scalar and vector latent quantization approaches, we consider bits per element, $Q \in \{4, 5, 6, 7\}$, while the block size $b$ is set to $2$ for the vector quantization scheme. For the neural compression approach, different models are trained corresponding to $\lambda = \{10^2, 5\times 10^2, 10^3, 5\times 10^3\}$. The block scaling parameter, $N_s$ is set to $320$ for all these schemes. In other word, only one scaling factor, $t$, is used for the entire input sequence.  
Two conventional schemes are considered as benchmarks. To be precise, we implement the scalar quantization (SQ) proposed in \cite{scalar_cpri} and the vector quantization (VQ) scheme from \cite{vq_cpri}. For a fair comparison, we assume they adopt the same $N_s$ as in the proposed schemes.  Their relative R-D performance in terms of EVM defined in \eqref{eq:evm_fd} and \eqref{eq:evm_fd_dB} are shown in Fig. \ref{fig:final_simu} (a) and (b), respectively.

As can be seen, the scalar quantization scheme in \cite{scalar_cpri} yields the worst performance as the elements of the input sequence are non-i.i.d.. Both the vector quantization schemes, i.e., the learned scheme as well as that in \cite{vq_cpri} adopt entropy coding illustrated in \eqref{eq:eq_cr} to further enhance the compression efficiency bringing an approximately $0.1$ bit per element reduction. Both the learned scalar and vector latent quantization approaches achieve better performance compared with the conventional scalar and vector quantization schemes, respectively. The learned latent vector quantization outperforms its uniform quantization counterpart as it relaxes the assumption that each element follows an i.i.d. uniform distribution (while the conventional one adopts non-uniform quantization).
Finally, we find the neural compression algorithm outperforms all the other schemes as expected, as it dedicatedly models the latent distribution and arithmetically encodes the quantized latent using a variable-rate codebook based on its distribution.

\subsubsection{Uplink case}
For the uplink scenario, we assume that a single user\footnote{It can be easily extended to the multi-user scenario, yet we will not show the corresponding simulations due to the page limit.} is communicating with the base station using the SC-FDMA modulation. A $7$-tap multi-path fading channel is considered where the coefficient of each tap is modeled as a complex Gaussian random variable drawn from $\mathcal{CN}(0, 1)$. The SNR is set to $5$ dB and the modulation order is $64$.

Different from the downlink case, after receiving the uplink signal, the RRH compresses the \textbf{entire} signal including CP as it is used to perform synchronization at the BBU. We assume the CP length equals to $64$ thus the input $\bm{s}$ to the latent encoding function is of $720$ dimensions resulting in a block scaling parameter $N_s = 360$, i.e., a single scaling factor $t$ for the whole input. 
The relative performance of the proposed learning-based schemes as well as the baselines is shown in Fig. \ref{fig:final_simu} (c). The latent vector quantization scheme achieves superior EVM performance over the conventional VQ scheme. However, the learned latent uniform quantization scheme is outperformed by the conventional SQ scheme when $Q < 6$, which is because the learned scheme assumes uniform quantization whereas the conventional SQ adopts non-uniform quantization.   
These observations are slightly different from the downlink case which is due to the fact that the pdf of the input signal $\bm{s}$ changes after passing through the multi-path fading channel. The proposed neural compression scheme, on the other hand,  shows its effectiveness which outperforms the conventional VQ by $Q \approx 0.3$ when EVM equals to $32$ dB.

\subsection{Performance with mismatched scenarios}
We then provide simulation results for the mismatched scenarios.  In particular, mismatched SNR values, constellation orders and number of channel taps are considered for the uplink scenario. Note that, since the neural compression model outperforms its latent quantization counterpart, we will focus on the neural compression scheme in the following simulations.

\subsubsection{Mismatched SNR values} In this simulation, we adopt the previous setting where 64QAM symbols are transmitted over a noisy 7-tap multi-path fading channel. Different models are trained with three different SNR values, namely, $\{-5, 5, 15\}$ dB. We evaluate the model trained at $5$ dB under SNR values of $\{-5, 15\}$ dB and the R-D curves are shown in Fig. \ref{fig:fig_mismatch} (a).

As can be seen, the proposed neural compression module shows its robustness against the SNR variations as the R-D curves evaluated at $\mathrm{SNR} = \{-5, 15\}$ dB are similar or slightly worse compared to the matched case, where the model is trained and tested under the same SNR value. 
% \footnote{Note that the R-D curve of the scheme trained at $\mathrm{SNR} = 5$ dB and tested at $\mathrm{SNR} = 15$ dB is not presented, because they are too close to distinguish.}
It is also observed that for the matched cases, the R-D performance under a large SNR value is better. This can be understood by considering the extreme cases: when $\mathrm{SNR} = \infty$ where no noise is added, the signal exhibits certain patterns and the neural compression module may utilize such patterns for better compression. On the other hand, when $\mathrm{SNR}$ is sufficiently low, the signal follows a Gaussian distribution, since Gaussian random variable has the largest entropy with a fixed variance, the R-D curve would be less satisfactory. 

\subsubsection{Mismatched modulation orders} The previous simulations assume $64$QAM symbols. We consider different constellation orders such as 4QAM and 16QAM and evaluate the model trained at $64$QAM on the mismatched constellations. The simulation assumes a 7-tap multi-path fading channel with $\mathrm{SNR} = 5$ dB. We found by experiments that the proposed scheme can achieve almost identical R-D performance regardless of the modulation order. Furthermore, the model trained at a specific constellation can achieve nearly the same performance when the modulation order changes. The simulation result is omitted due to the page limit.

\subsubsection{Mismatched number of taps}
In this part, we demonstrate the robustness of the proposed scheme by changing the number of channel taps, denoted as $N_{tap}$, of the noisy multi-path fading channel. In particular, the SNR of the channel is fixed at $5$ dB and we consider $N_{tap} = \{1, 3, 7\}$. Then, the R-D performance with respect to different $N_{tap}$ is shown in Fig. \ref{fig:fig_mismatch} (b). 
It can be seen that a better R-D curve is obtained when $N_{tap}$ is larger.
%\textcolor{blue}{This can be understood in the frequency domain where the signal of interest is multiplied with the frequency domain response. For $N_{tap} = 1$, its frequency response would be spread over all the subcarriers. When $N_{tap}$ becomes larger, some of the subcarriers may experience deep fade and thus the number of `effective' subcarriers would be reduced. }
Moreover, the proposed scheme can achieve similar or slightly worse performance when there is a mismatch in the number of hops.

\begin{figure}
    \centering
\begin{subfigure}{\columnwidth}
    \centering
    \includegraphics[width=0.8\columnwidth]{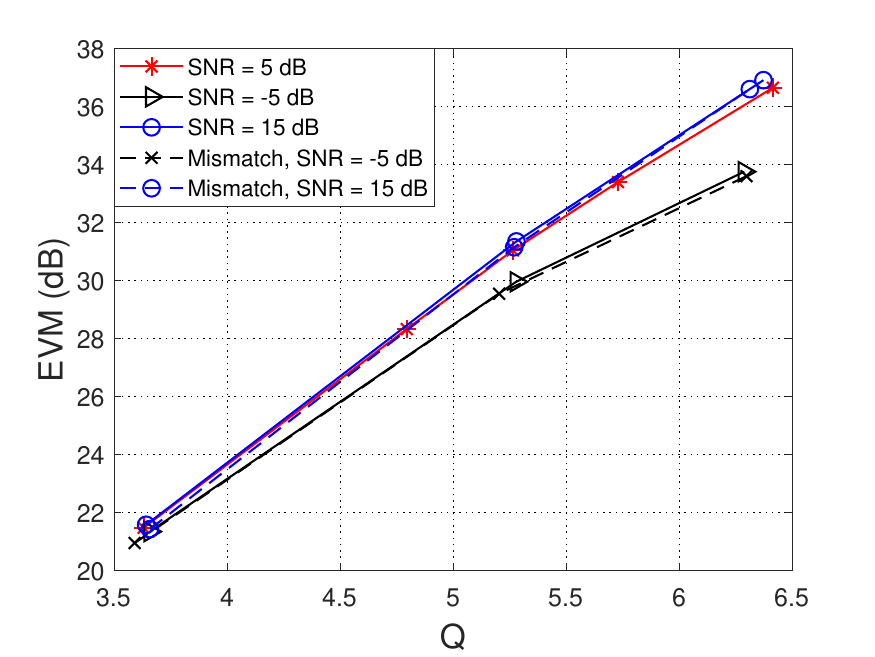}
    \caption{Mismatched SNRs}
\end{subfigure}

\begin{subfigure}{\columnwidth}
    \centering
    \includegraphics[width=0.8\columnwidth]{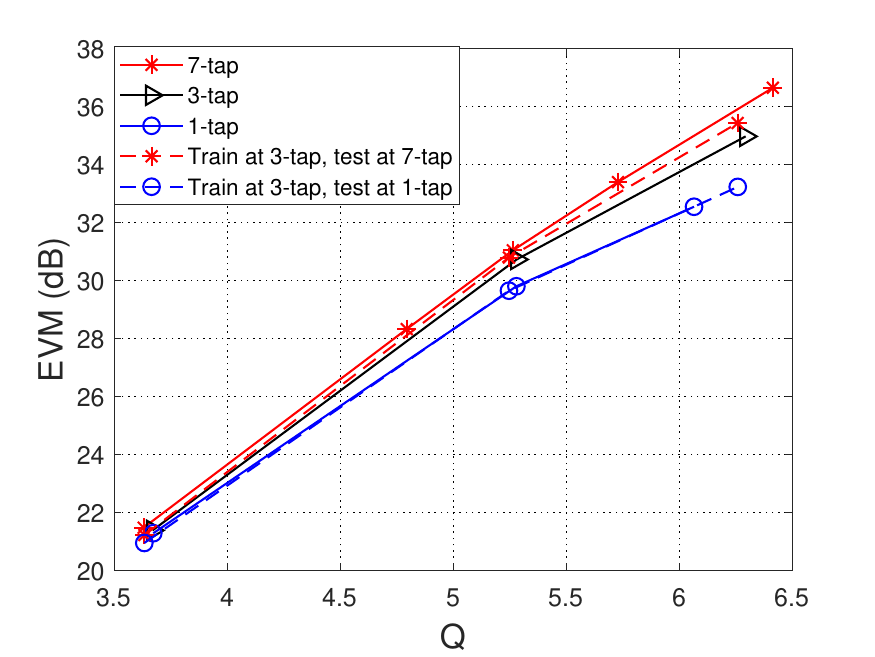}
    \caption{Mismatched $N_{tap}$}
\end{subfigure}

\caption{{The R-D curves for the mismatched scenarios, namely, the mismatched SNR values and the number of channel taps. The dashed lines denote the R-D curve when a model trained with a specific SNR/number of hops is evaluated under a different setting.}}
\label{fig:fig_mismatch}
\end{figure}

\subsection{Improved Neural Compression Solutions}
\subsubsection{Successive refinement model}
We then show the performance of the successive refinement model illustrated in Section \ref{sec:successive_refine}. In this simulation, the same setup in Fig. \ref{fig:final_simu} (c) is considered and the maximum number of layers is set to 3, i.e., $L = 3$. The base layer is trained with $\lambda_1 = 10^2$ and the subsequent layers with indexes $\ell = 2, 3$ are trained with carefully chosen $\lambda_2, \lambda_3 = 10^3, 10^4$, respectively. It can be seen from Fig. \ref{fig:fig_sr_sim} that the successive refinement scheme has a small gap with respect to the separately trained baselines. This is due to the fact that the successive refinement module requires reconstruction at the $\ell$-th layer after receiving $\{\bm{b}_1, \ldots, \bm{b}_\ell\}$. The separately trained models on the other hand, reconstruct the model after all the bit sequence, $\bm{b}$ is received. A more stringent constraint results in a performance loss compared to the separately trained ones.

%\textcolor{blue}{Finally, we give concrete example to illustrate the merit of the successive refinement scheme. If the required reconstruction SQNR is 30 dB and both schemes start with the lowest reconstruction level with $\lambda = \lambda_1 = 100$, then the successive refinement scheme can finish transmission with $Q^*\approx5.5$ where $Q^*$ denotes the overall bpp to achieve the target. The learned scheme, on the other hand would require at least $Q^* = 3.5 + 5 \approx 8.5$. }

\begin{figure}[!t]
\centering
\includegraphics[width=0.8\columnwidth]{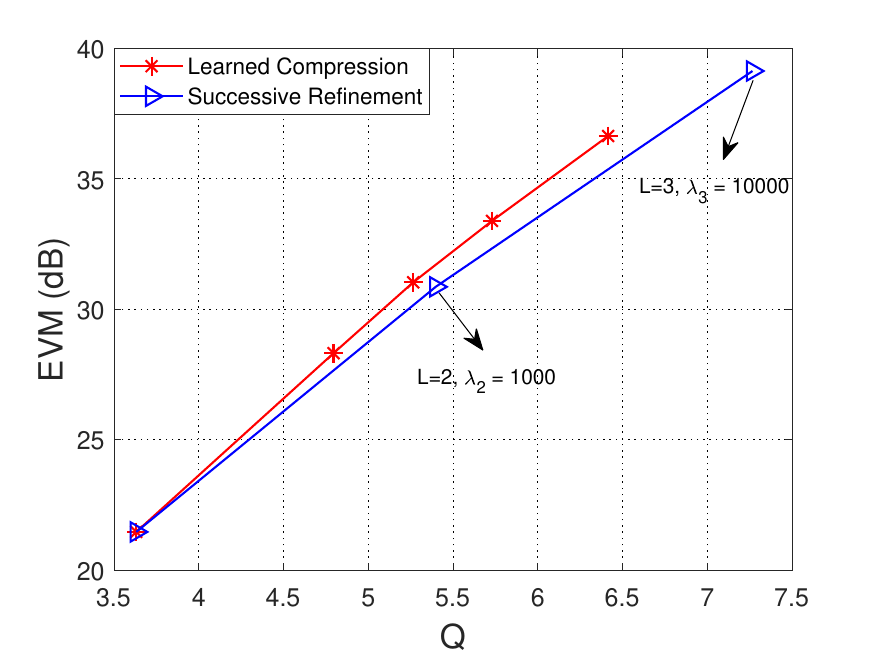}
\caption{The relative R-D performance between the successive refinement models and the separately trained ones. We consider $L = 3$ with $\lambda_\ell \in \{10^2, 10^3,10^4\}$ for the successive refinement model.}
\label{fig:fig_sr_sim}
\end{figure}

\begin{figure}[!t]
\centering
\includegraphics[width=0.8\columnwidth]{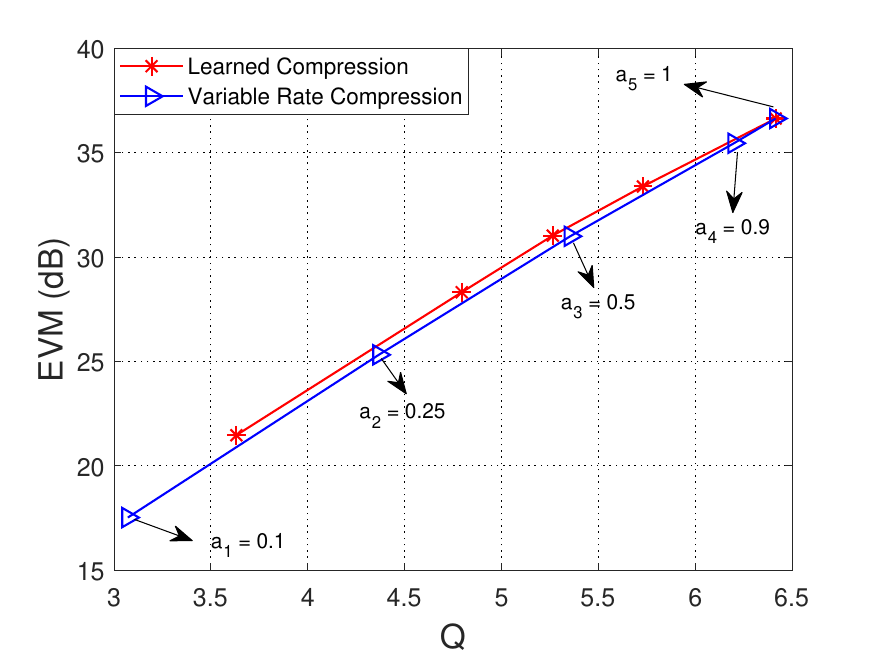}
\caption{The relative performance between the variable-rate compression model with shared non-linear transformation functions, $f_s^*(\cdot), g_d^*(\cdot)$ and the models trained with distinct $f_s(\cdot), g_d(\cdot)$.}
\label{fig:fig_vr_rdcurve}
\end{figure}

\begin{figure}[!t]
\centering
\includegraphics[width=0.8\columnwidth]{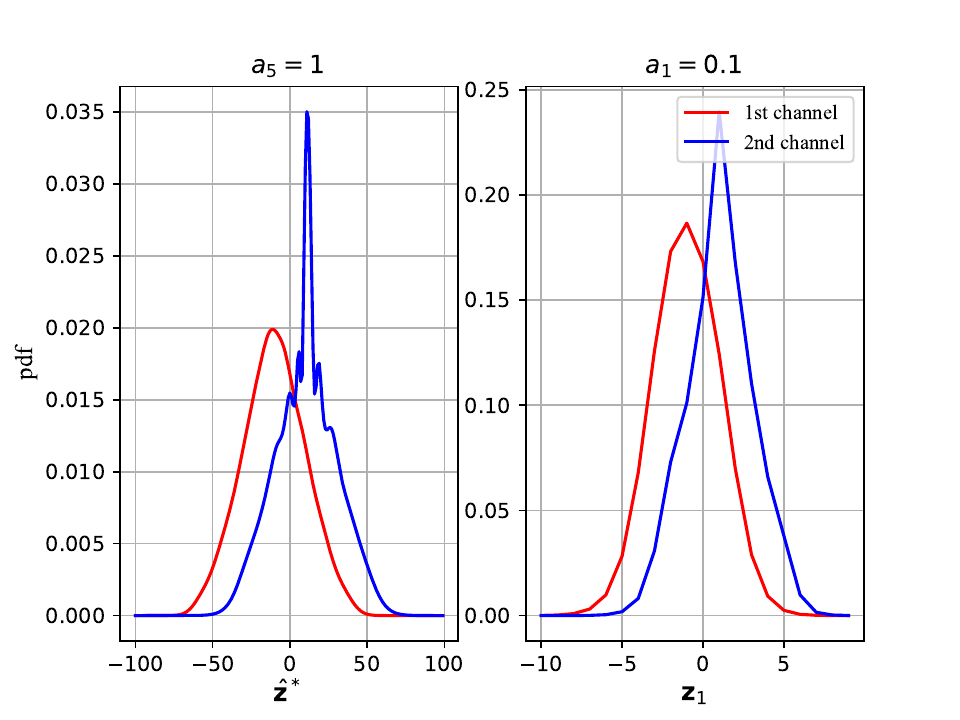}
\caption{We plot the entries of the codebook $\mathcal{C}_W$ as well as the codebook $\mathcal{C}_1$ corresponding to $a_W = 1$ and $a_1 = 0.1$, respectively. Both the codebooks have two rows/channels and there is a shrinkage in the support set for the $a_1 = 0.1$ case.}
\label{fig:fig_pdf_plot}
\end{figure}

\subsubsection{Variable-rate compression}
Finally, we manifest the effectiveness of the proposed variable-rate compression model. The uplink scenario with 7-tap noisy multi-path fading channel and $\mathrm{SNR} = 5$ dB is considered. The shared latent encoding and decoding functions, $f_s^*(\cdot), g_d^*(\cdot)$, are selected from the model trained with $\lambda = 5 \times 10^3$. $W = 5$ scaling factors are introduced, namely, $a_w \in \{0.1, 0.25, 0.5, 0.9, 1\}$. The R-D performance of the variable-rate compression model with respect to the separately trained models in Fig. \ref{fig:final_simu} (c) is shown in Fig. \ref{fig:fig_vr_rdcurve}.

As can be seen, the variable-rate compression model can achieve similar reconstruction performance with reduced storage cost compared to the separately trained models. 
In particular, the two curves coincide at $a_W = 1$ and a small gap of $Q \approx 0.1$ is observed when EVM equals to 25 dB. We then quantify the reduction in storage for the variable-rate compression model. For each of the separately trained model, the weights of the latent encoding and decoding functions, $f_s(\cdot)$ and $g_d(\cdot)$ are stored together with the corresponding codebook, $\mathcal{C}$. As a concrete example, approximately $2.5\times10^4$ floating-point parameters are introduced for the two-layer bi-directional GRU as well as the FC model shown in Fig. \ref{fig:fig_NN}. 
In contrary, the codebook $\mathcal{C}$ (or equivalently, $\mathcal{C}_W$), for the model with $\lambda = 5 \times 10^3$ is only comprised of $2\times200$ entries.  We note that each entry corresponds to a non-zero probability value and for each of the $V=2$ channels, all the corresponding entries summed up to be 1. Due to the fact that the size of the codebook is negelectable compared to the weights of the non-linear transformation functions, the proposed scheme reduces the storage cost to $1/W$ compared to the previous scheme.

We further provide the learned pdf, i.e., the entries of the codebook, $\mathcal{C}_w$, of the scaled latent vector, $\bm{Z}_w$. In particular, the entries of $\mathcal{C}_W$ and $\mathcal{C}_1$ are shown in  Fig. \ref{fig:fig_pdf_plot}. It is clear from the figure that the support set of $\mathcal{C}_1$ shrinks approximately $1/10$ compared to that of the $\mathcal{C}_W$. This matches with the analysis in Section \ref{sec:variable_rate}, that the entropy of $\bm{Z}_w$ reduces after scaling. The figure also shows that two different channels for the codebook follow two different distributions.

\section{Conclusion}
This paper proposed two AI-native fronthaul compression algorithms: latent quantization and neural compression, for efficient signal transmission over the CPRI link in C-RAN architectures. These methods extend traditional decimation and block scaling frameworks through learned non-linear transformations and latent source coding strategies, offering substantial improvements in rate-distortion performance.
{To address practical deployment scenarios in NextG networks, we further introduced two system-level enhancements: a successive refinement model for adaptive transmission over heterogeneous CPRI links and a variable-rate compression framework using shared neural components to reduce the storage on resource-constrained devices. 
Extensive evaluations confirmed that the proposed models consistently outperform conventional compression schemes in both uplink and downlink settings, and remain robust under mismatched channel conditions, SNRs, and modulation formats offering valuable implications for standardization efforts in 6G.
}

\appendices
\section{On the distribution of the time domain signal}\label{sec:APPA}
\begin{figure}[!t]
\centering
	\begin{subfigure}{0.7\linewidth}
		\centering
		\includegraphics[width=\linewidth]{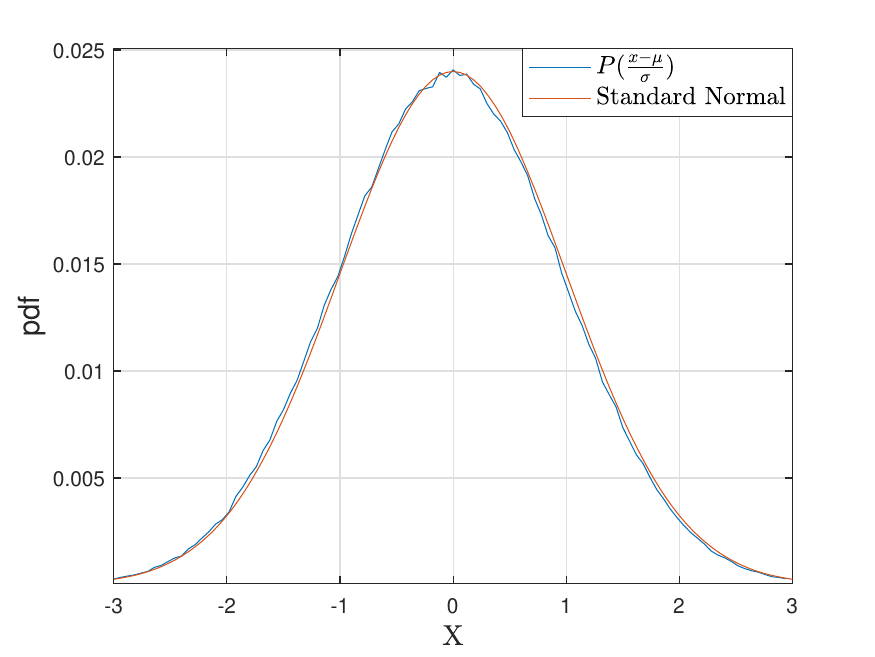}
		\caption{}
	\end{subfigure}

	\begin{subfigure}{0.7\linewidth}
		\centering
		\includegraphics[width=\linewidth]{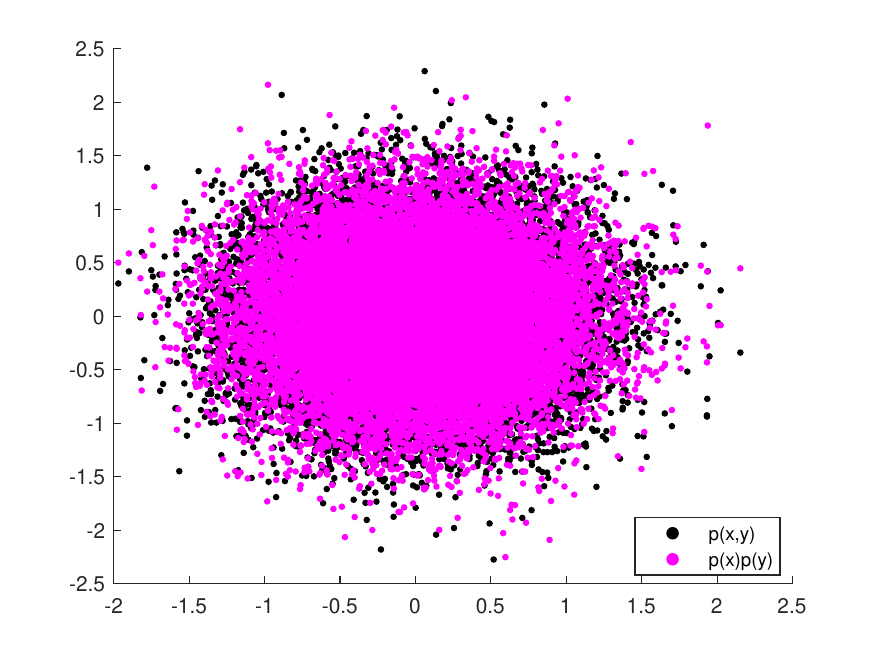}
		\caption{}
	\end{subfigure}

\caption{Visualization of the probability distribution for the random variables, $X$ and $Y$: (a) the marginal distribution of $\frac{X-\mu}{\sigma}$ deviates from the standard normal distribution; (b) different variables, $X$ and $Y$ are correlated which is verified by comparing the joint distribution $p(x, y)$ with the product of the marginal distributions, $p(x)p(y)$.}
\label{fig:fig_pdf_scatter}
\end{figure}

This appendix provides detailed derivation of the distribution of the decimation module output, $\bm{x}^\prime$, shown in Proposition \ref{prop:1} of Section \ref{sec:latent_quantize}.

We are interested in exploring whether the elements of the input signal, $\bm{s}$, which is the input to the transformation and source encoding modules, are i.i.d.. In the downlink case, $\bm{s}$ is obtained by first applying IFFT to $N_{fft}$ subcarriers where $N_{sym}$ of them are occupied by i.i.d. QAM symbols, denoted by $\bm{c}$, followed by decimation and block scaling.
%For the uplink case, the input signal is obtained after an extra multi-path fading channel. 
We will show for the downlink case\footnote{It is easy to show the same property for the uplink scenario.} that the elements in $\bm{s}$ are non-i.i.d. which motivates the introduction of the non-linear transformation functions in Section \ref{sec:latent_quantize}. 

We first focus on the output signal of the decimation module, $\bm{x}^\prime$. Note that we carefully select the $K$ and $M$ values\footnote{We have the relation $\frac{N_{sym}}{N_{fft}} < \frac{K}{M}$.} and design a low-pass filter such that the spectrum of $\bm{x}^\prime$, denoted as $\bm{w}$, is still comprised of the same  QAM symbols, $\bm{c}$, but with a narrower guard band (in other words, less number of zero subcarriers) compared to the spectrum of $\bm{x}$. In particular, $\bm{w}$ can be expressed as:
\begin{equation}
    \bm{w} = [\bm{0}_{k^*}^\top, \bm{c}^\top, \bm{0}_{k^*}^\top]^\top,
\end{equation}
where\footnote{We assume $(N^\prime - N_{sym})$ is an even number which can be satisfied via zero padding of the QAM symbols, $\bm{c}$.} {$k^* = (N^\prime - N_{sym})/2$}. The $n$-th element of $\bm{x}^\prime$ can then be expressed as:
\begin{equation}
    {x}^\prime_n = 
\frac{1}{\sqrt{N^\prime}}\sum_{k=k^*}^{N^\prime-k^*-1}{w}_k e^{j \frac{2\pi kn}{N^\prime}}.
\label{eq:app_xn_prime}
\end{equation}
\begin{comment}
    We study the real part (the analysis of the imaginary part follows the same) of ${x}^\prime_n$, denoted as $\Re({x}^\prime_n)$:
\begin{equation}
    \Re({x}^\prime_n) = \frac{1}{\sqrt{N^\prime}}\sum_{k=k^*}^{N^\prime-k^*-1}{w}^r_k \cos(\frac{2\pi nk}{N^\prime}) - {w}^i_k \sin(\frac{2\pi nk}{N^\prime}),
    \label{eq:real_x}
\end{equation}
\end{comment}
The expectation of ${x}^\prime_n$ is expressed as:
\begin{align}
    \mathbb{E}({x}^\prime_n) = \frac{1}{\sqrt{N^\prime}}\sum_{k=k^*}^{N^\prime-k^*-1}\mathbb{E}\Big({w}^r_k + j {w}^i_k\Big) e^{j\frac{2\pi nk}{N^\prime}}  = 0,
    \label{eq:eq_mu}
\end{align}
where ${w}^r_k$ and ${w}^i_k$ represent the real and imaginary parts of the QAM symbol, ${w}_k$, respectively. Since ${w}^r_k$ and $ {w}^i_k$ are i.i.d. and can be seen as pulse-amplitude modulation (PAM) symbols, we have $\mathbb{E}({w}^r_k) = \mathbb{E}({w}^i_k) = 0$. 
The covariance between ${x}^\prime_n$ and ${x}^\prime_m$ can be derived using \eqref{eq:eq_mu} and the fact that different ${w}_k$ elements are i.i.d.:
\begin{align}
    \text{Cov}({x}^\prime_n, {x}^\prime_m) &= \frac{1}{{N^\prime}}\mathbb{E}\sum_{k_1}\sum_{k_2}{w}_{k_1} e^{j\beta n k_1}{w}^\dagger_{k_2} e^{-j\beta m k_2} \notag \\
    &=\frac{P}{{N^\prime}}\sum_{k=k^*}^{N^\prime-k^*-1} e^{j\beta (n-m) k},
\end{align}
where we define $\beta \triangleq \frac{2\pi}{N^\prime}$, $\mathbb{E}(|w_k|^2) \triangleq P$ and utilize the fact that $\mathbb{E}(w_{k_1}^\dagger w_{k_2}) = P \delta_{k_1, k_2}$. 
\begin{comment}
    The mean and variance of the imaginary part can be derived in the same way:
\begin{align}
    \mathbb{E}(\Im(\bm{x}^\prime_n)) &= 0, \notag \\
    \text{Cov}(\Im(\bm{x}^\prime_n), \Im(\bm{x}^\prime_m)) &= \frac{P}{{N^\prime}}\sum_{k=k^*}^{N^\prime-k^*-1} \sin(\beta (n-m) k).
\end{align}
\end{comment}
It is easy to see when $k^* = 0$, the covariance would always equal to zero unless $m = n$. However, in practical settings, $k^*$ is non-negative to ensure that the low-pass filter preserves the signals located at the occupied subcarriers. Thus, the elements of $\bm{x}^\prime$ are correlated (non-i.i.d.) and the subsequent block scaling operation does not change the property. This finishes the proof of Proposition \ref{prop:1} validating that transformations are needed to achieve a better CPRI compression performance.

\begin{comment}
    \textcolor{blue}{\textbf{\textit{Discussion}}: A seemingly plausible approach is to perform FFT to transform $\bm{s}$ to the spectrum domain and directly compress the signal at the occupied subcarriers. This approach removes all the redundancy introduced by the guard band and is also computational efficient. However, we point out this is impractical as the RRH/BBU may not be able to perform perfect synchronization. If the time-domain signal is misaligned with the original one, there would lead to a significant distortion in the frequency domain. On the other hand, the decimation block whose core component is a low-pass filter, is much more tolerate to the synchronization error thus we adopt the structure shown in Fig. \ref{fig:flowchart}.}
\end{comment}

We further analyze the marginal distribution of the $n$-th element, ${x}^\prime_n$ and the joint distribution of two different elements, denoted as ${x}^\prime_n$ and ${x}^\prime_m$. We study the real part\footnote{The distribution of the imaginary part follows the same analysis.} of ${x}^\prime_n$, which is a random variable, denoted by $X$. According to the central limit theory, when $N_{sym}$ is large, the average of i.i.d. random variables, $w_k$ in \eqref{eq:app_xn_prime}, would converge to a normal distribution, $\mathcal{N}(\mu, \sigma^2)$. In the considered scenario, however, due to the weighted term, i.e., $e^{j\frac{2\pi nk}{N^\prime}}$, the distribution of  $X$ deviates slightly from the normal distribution which can be derived using the characteristic function.
Due to the page limit, we simply illustrate this by comparing the empirical distribution of $\frac{X-\mu}{\sigma}$, denoted as $P\Big(\frac{X-\mu}{\sigma}\Big)$ with the standard normal distribution in Fig. \ref{fig:fig_pdf_scatter} (a). Note that sufficient number of $X$ implementations are collected to estimate $P\Big(\frac{X-\mu}{\sigma}\Big)$.

We then illustrate the dependency of the real parts of the two random variables, ${x}^\prime_n$ and ${x}^\prime_m$, which we denote them as $X$ and $Y$. This is done by showing the difference between the joint distribution, denoted as $p(x,y)$ and the product of two marginal distributions, $p(x)p(y)$. We draw $2\times 10^4$ samples  from the two distributions which are as shown in Fig. \ref{fig:fig_pdf_scatter} (b). As can be seen, the samples drawn from the two marginal distributions, $p(x)p(y)$, are more concentrated to the origin compared with those drawn from the joint distribution, $p(x,y)$. This further validates the correlation among different elements.

\bibliographystyle{IEEEbib}
\bibliography{abs_ref}

\end{document}